\documentclass[pdflatex,sn-mathphys-num]{sn-jnl}
\usepackage{comment}
\usepackage{graphicx}%
\usepackage{amsmath,amssymb,amsfonts}%
\usepackage{amsthm}%
\usepackage{mathrsfs}%
\usepackage{xcolor}%
\usepackage{textcomp}%
\usepackage{booktabs}%
\usepackage{tabularx}%
\usepackage{float}
\IfFileExists{placeins.sty}{\usepackage{placeins}}{\newcommand{\FloatBarrier}{\clearpage}}%

\floatstyle{ruled}
\newfloat{algorithm}{tbp}{loa}
\floatname{algorithm}{Algorithm}
\makeatletter
\providecommand*{\toclevel@algorithm}{1}
\newcounter{ALG@line}
\newenvironment{algorithmic}[1][0]{%
  \setcounter{ALG@line}{0}%
  \def\ALG@freq{#1}%
  \begin{list}{\ifnum\ALG@freq>0\refstepcounter{ALG@line}\arabic{ALG@line}:\fi}{%
\setlength{\leftmargin}{2.2em}%
\setlength{\labelwidth}{1.8em}%
\setlength{\labelsep}{0.4em}%
\setlength{\itemsep}{2pt}%
\setlength{\parsep}{0pt}%
\setlength{\topsep}{3pt}%
  }%
}{\end{list}}
\newcommand{\Require}{\item[\textbf{Require:}]}
\newcommand{\Ensure}{\item[\textbf{Ensure:}]}
\newcommand{\State}{\item}
\newcommand{\For}[1]{\State \textbf{for} #1 \textbf{do}}
\newcommand{\EndFor}{\State \textbf{end for}}
\makeatother

\theoremstyle{thmstyleone}%
\newtheorem{theorem}{Theorem}[section]
\newtheorem{proposition}[theorem]{Proposition}

\newtheorem{corollary}[theorem]{Corollary}

\theoremstyle{thmstylethree}%
\newtheorem{assumption}[theorem]{Assumption}

\theoremstyle{thmstyletwo}%

\newcommand{\dd}{\mathrm d}
\newcommand{\E}{\mathbb E}
\newcommand{\Q}{\mathbb Q}
\newcommand{\R}{\mathbb R}
\newcommand{\F}{\mathcal F}

\providecommand{\footnoteA}[1]{\footnote{#1}}

\hypersetup{hypertexnames=false}

\makeatletter
\renewenvironment{table}[1][]{%
  \par\addvspace{12pt}%
  \noindent\begin{minipage}{\textwidth}%
  \def\@captype{table}%
  \centering%
  \tablebodyfont%
}{%
  \end{minipage}%
  \par\addvspace{12pt}%
}

\makeatother

\begin{document}

\title[A Hybrid LSMC--PDE Method for Bermudan Option Pricing under the Gatheral Double Mean-Reverting Model]{A Hybrid LSMC--PDE Method for Bermudan Option Pricing under the Gatheral Double Mean-Reverting Model}

\author[1]{\fnm{Mara Kalicanin} \sur{Dimitrov}}\email{mara.kalicanin.dimitrov@mdu.se}
\author*[1]{\fnm{Ying} \sur{Ni}}\email{ying.ni@mdu.se}

\affil*[1]{\orgdiv{Department of Business and Mathematics}, \orgname{M\"{a}lardalen University}, \orgaddress{\postcode{721~23}, \city{V\"{a}ster{\aa}s}, \country{Sweden}}}

\abstract{
We study Bermudan option pricing under the {empirically supported} Gatheral Double Mean-Reverting (GDMR) stochastic volatility model. The model features a variance process together with a stochastic long-run mean variance process and allows Constant Elasticity of Variance (CEV)-type exponents in the diffusion coefficients. 
Given that the pricing of early-exercise derivatives under the GDMR model remains largely unexplored in the literature, we adapt a hybrid least-squares Monte Carlo--partial differential equation (Hybrid LSMC--PDE) framework to this model and provide a detailed model-specific implementation. {We derive a conditional Gaussian representation, which also yields a semi-analytical representation for European option pricing under the GDMR model.} Conditioning on simulated variance paths, the pricing problem reduces to a one-dimensional problem in the asset price, which is solved by a Fourier-based approach, while the remaining dependence on the variance variables is approximated by least-squares regression. 
{Our numerical experiments demonstrate that the Hybrid LSMC--PDE method provides accurate pricing estimates under the GDMR model, often with lower pricing errors than plain least-squares Monte Carlo (plain LSMC) for low and moderate simulation budgets. Additional validation under the Heston model confirms the accuracy of the numerical implementation against established European and Bermudan option pricing results.}
}
\keywords{Bermudan option pricing, stochastic volatility, Hybrid LSMC--PDE, least-squares Monte Carlo, Gatheral model}

\maketitle

\section{Introduction}

Bermudan option pricing under stochastic volatility requires repeated
approximation of continuation values defined as conditional expectations of discounted future payoffs. Deterministic Partial Differential Equation (PDE) methods handle this recursion well in low-dimensional Markovian models, but their cost grows quickly when several
volatility factors are present. Under a multi-factor stochastic volatility model, the resulting three or higher dimensional pricing problem becomes increasingly computationally demanding.

Regression-based Monte Carlo approaches for pricing American-style options were developed by Tsitsiklis and Van Roy~\cite{tsitsiklis2001regression} and later popularized through the Least-Squares Monte Carlo (LSMC) method of Longstaff and Schwartz~\cite{longstaff2001valuing}, which is now widely used for pricing options with early exercise features.
Clément, Lamberton, and Protter~\cite{clement2002analysis} analyzed the
convergence of the least-squares regression method for American option pricing.
These methods are flexible, but in stochastic-volatility models the regression approximation of the continuation value function must capture the dependence on both asset and
volatility variables.

Mixed Monte Carlo--PDE methods, which combine simulation with conditional PDE calculations, have shown to be promising for option pricing problems under stochastic-volatility models; see, for example, Loeper and Pironneau~\cite{loeper2009mixed}, Lipp, Loeper, and Pironneau~\cite{lipp2013mixing}, and Cozma and Reisinger~\cite{cozma2017mixed}. More recently, Farahany, Jackson, and Jaimungal~\cite{farahany2020mixing} developed a Hybrid LSMC--PDE method for pricing Bermudan options under stochastic volatility, which forms the methodological basis of the present work. 
Their algorithm simulates the volatility process, solves conditional
asset-direction PDE problems along simulated volatility paths, and then
regresses the resulting conditional values over the volatility state. {They also establish almost-sure convergence of the regression coefficients for a class of stochastic-volatility models with one-way coupling under certain technical assumptions.}

We study Bermudan option pricing under the Gatheral-type
double mean-reverting stochastic volatility (GDMR) model,
introduced in Gatheral~\cite{gatheral2008}, where successful
simultaneous calibration to both SPX and VIX option markets
was demonstrated. {The GDMR model is a stochastic-volatility model with one-way coupling, making it suitable for the Hybrid LSMC--PDE framework.} Under this model, the instantaneous
variance $v$ mean-reverts to a stochastic variance level
$v'$, while $v'$ mean-reverts to a constant long-run level.
{In contrast to the Heston stochastic volatility model ~\cite{heston1993closedform}, where
variance mean-reverts to a constant level, the stochastic
mean-reversion level in the GDMR model can evolve with changing economic conditions. The Heston model is recovered as a special case of the GDMR model when the mean-reversion level is constant.  Further empirical support is provided by Bayer, Gatheral, and Karlsmark~\cite{bayer2013fast},
who demonstrated good simultaneous fits to SPX and VIX option prices under different market conditions, with relatively stable calibrated parameters. Huh, Jeon, and Kim~\cite{huh2018scaled} compared the
Heston, double mean-reverting, and scaled double mean-reverting models using VIX futures and options data. Their results provide further evidence of improved calibration accuracy relative to the Heston model, particularly during periods of market turbulence.}

Despite its appealing modelling properties, the GDMR model does not admit closed-form pricing formulas, even for European options. As a result, calibration in Gatheral~\cite{gatheral2008} was performed using full Monte Carlo simulation with an Euler discretization scheme, while Bayer, Gatheral, and Karlsmark~\cite{bayer2013fast} later employed the higher-order Ninomiya--Victoir scheme to improve computational efficiency. Moreover, there has been limited research on Bermudan and American option pricing under this framework. To the best of our knowledge, Haida et al.~\cite{haida2026almost} is among the few studies addressing this problem, proposing a simulation scheme together with a full Least-Squares Monte Carlo (LSMC) method for American option pricing under a double-Heston-type parameterization of the GDMR model. In this work, we aim to fill this gap by developing a Hybrid LSMC--PDE method, adapted from the more general framework of Farahany, Jackson, and Jaimungal~\cite{farahany2020mixing}. This approach is motivated by the computational challenges of applying a full PDE method in the presence of three state variables, while avoiding the need for a full regression approximation of the continuation value as in standard LSMC {(hereafter referred to as plain LSMC)}. Although our discussion focuses on Bermudan option pricing, the framework naturally extends to early-exercise derivatives more broadly through Bermudan approximations of American options.

We note two advantages of this hybrid method relative to the plain LSMC method.  First, rather than approximating the entire continuation value through a parametric regression approximation, the Hybrid LSMC--PDE method performs regression only partially, leaving the stock-price dependence to be resolved through a conditional pricing PDE. {Second, the conditional PDE step naturally produces an option price curve over a grid of stock prices instead of a single option price as in the plain LSMC method. From the resulting option price curve, the corresponding Delta curve, representing one of the most important option Greeks, can be obtained by numerical differentiation at negligible additional computational cost.} 

The present paper focuses on adapting the Hybrid LSMC--PDE framework to the GDMR model, motivated by the considerations discussed above. This adaptation is nontrivial for the Gatheral double mean-reverting model, because the model has two correlated variance Brownian drivers and Constant Elasticity of Variance (CEV)--type variance coefficients. Our contribution is thus a model-specific formulation and implementation of the existing Hybrid LSMC--PDE framework for the multifactor GDMR setting. {To adapt this framework to the GDMR model, we derive the Brownian projection and the conditional one-step representation (Proposition~\ref{prop:conditional-gaussian-step}).} The key model-specific quantities are the correlated Brownian shift $Z_n$, the integrated drift $I_n$, the integrated residual variance $B_n$, and the terminal volatility state. {As a by-product, we also obtain a corollary for European option pricing under the GDMR model (Corollary~\ref{cor:european}).}

{From a computational perspective, we provide a  detailed description of the numerical implementation of the Hybrid LSMC--PDE method, based on the fast Fourier transform (FFT), for Bermudan option pricing under the GDMR model. Moreover, we conduct extensive numerical experiments focusing on Bermudan option pricing and additional validation using European option prices. Our numerical experiments investigate pricing accuracy under the GDMR model, computational efficiency when evaluating option prices across a range of initial stock prices, and additional validation under the Heston specialization of the GDMR model,  using established European and Bermudan option pricing results. The experiments cover three parameter sets from the literature, including two market-calibrated sets, and a range of moneyness
levels, maturities, exercise frequencies and simulation settings. Overall, the experimental results suggest that the Hybrid LSMC--PDE method provides a competitive alternative to conventional Monte Carlo-based approaches for Bermudan and, as a special case,
European option pricing.} 

In addition, we also include a short well-posedness result for the variance subsystem of the GDMR dynamics. This result is not the main contribution of the paper, it is included to make explicit the closed nonnegative variance domain used in the Bermudan recursion and in the conditional construction.

The remainder of the paper is organized as follows. Section~\ref{sec:model}
discusses the GDMR model and gives the
well-posedness result. Section~\ref{sec:conditional_pde} gives the theoretical framework for Bermudan option pricing using the Hybrid LSMC--PDE method and, in particular, derives the
conditional Hybrid LSMC--PDE representation. {It also presents a corollary for European
option pricing under the GDMR model.}  Section~\ref{sec:hybrid-method} describes the numerical implementation of the Hybrid LSMC--PDE method for Bermudan option pricing.  Section~\ref{sec:numerical} reports the numerical experiments under the GDMR model, {while Section~\ref{sec:heston-validation} presents additional validation under the Heston specialization of the GDMR model.} Section~\ref{sec:conclusion} concludes.

\section{Gatheral double mean-reverting model}\label{sec:model}

This section introduces the GDMR model. In this model, the equations for the variance factors do not involve the asset price, while
the asset equation depends on the variance factors. This one-way coupling is
used in Section~\ref{sec:conditional_pde} to obtain the conditional
one-dimensional representation.

\subsection{GDMR dynamics}

Fix a finite horizon $T>0$. Let
$$
(\Omega,\F,(\F_t)_{t\in[0,T]},\Q)
$$
be a filtered probability space satisfying the usual conditions.  Throughout the paper, option pricing is performed under a risk-neutral (equivalent martingale) measure $\Q$ under which the discounted stock price process is assumed to be a true martingale. We assume a deterministic money-market account $B_t=e^{rt}$ where $r$ is the constant risk-free interest rate. Unless otherwise specified, all expectations are taken with respect to the probability measure $\Q$. Furthermore, the notation $\E_\xi[\cdot]$ denotes expectation with respect to the probability distribution of the random variable $\xi$.

The probability space supports a three-dimensional correlated $(\F_t)$-Brownian motion
$$
W=(W^{(1)},W^{(2)},W^{(3)})
$$
with constant quadratic covariations
$$
\langle W^{(i)},W^{(j)}\rangle_t=\rho_{ij}t,\qquad 1\le i,j\le 3.
$$
{
The matrix $\rho=(\rho_{ij})_{1\le i,j\le 3}$ is assumed to be symmetric positive definite and satisfies $\rho_{ii}=1$. This ensures that the covariance matrix of the two variance Brownian drivers, $W^{(2)}$ and $W^{(3)}$, is nonsingular. Moreover, the decomposition of $W^{(1)}$ into its projection onto the span of $(W^{(2)},W^{(3)})$ and an orthogonal residual Brownian component, introduced in
Section~\ref{sec:conditional_pde}, is nondegenerate.}

We consider the GDMR dynamics of Gatheral~\cite{gatheral2008}, augmented with a risk-free drift term in the stock-price dynamics:
\begin{equation}\label{eq:gdmr}
\begin{aligned}
\dd S_t &= r\,S_t\,\dd t + S_t\sqrt{v_t}\,\dd W^{(1)}_t,\\
\dd v_t &= \kappa_1\,(v'_t - v_t)\,\dd t + \xi_1 v_t^{\delta_1}\,\dd W^{(2)}_t,\\
\dd v'_t &= \kappa_2\,(\theta - v'_t)\,\dd t + \xi_2 (v'_t)^{\delta_2}\,\dd W^{(3)}_t,
\end{aligned}
\end{equation}
where $r\ge0$, $\kappa_1,\kappa_2,\theta,\xi_1,\xi_2>0$, and
$\delta_1,\delta_2\in[1/2,1]$. The factor $v$ is the instantaneous variance and
mean-reverts to the stochastic level $v'$, while $v'$ mean-reverts to the
long-run level $\theta$. The exponents $\delta_1$ and $\delta_2$, referred to as \emph{CEV exponents}, replace the
square-root volatility-of-volatility coefficients by CEV--type powers; the
square-root double mean-reverting model is recovered when
$\delta_1=\delta_2=1/2$.

We further restrict the leverage correlations to satisfy 
$
\rho_{12},
\rho_{13} \in \mathopen{{(}}-1,0].
$
In particular, negative stock-volatility correlations are known to mitigate potential moment explosion effects and support stable martingale behavior in stochastic volatility models as pointed out in Andersen and Piterbarg~\cite{AndersenPiterbarg2007}. This assumption is also consistent with market stylized facts on stock-volatility correlations.



\subsection{Well-posedness and nonnegativity}\label{sec:wellposed}

The purpose of this subsection is to state the state-space property of the variance dynamics used later in the pricing construction. Since the variance factors are simulated first and then used in the conditional Hybrid LSMC--PDE step, the factors $v$ and $v'$ must be well defined and nonnegative on the exercise horizon.

The only delicate point is the boundary at zero. When $\delta_i<1$, the CEV coefficients are not Lipschitz at the boundary, so the usual globally Lipschitz SDE argument does not apply directly. This is why we state the well-posedness and nonnegativity result explicitly.

Mishura, Pilipenko, and Ralchenko~\cite{mishura2025gatheral} also study the Gatheral double stochastic-volatility system. In their notation, the variance factors are $X$ and $Y$, corresponding to $v$ and $v'$ in \eqref{eq:gdmr}, with $a_1=b_1=\kappa_1$, $a_2=\kappa_2\theta$, $b_2=\kappa_2$, $\sigma_i=\xi_i$, and $\alpha_i=\delta_i$. They assume strictly positive deterministic initial values and prove strong existence, uniqueness, nonnegativity, and the strong Markov property. Their paper then studies a different issue: whether the volatility factors can remain close to zero, and how a Skorokhod-reflected modification of the internal factor can prevent this behaviour.

We keep the risk-neutral GDMR dynamics in Eq.~\eqref{eq:gdmr}, and we use the result only to justify the closed nonnegative variance domain needed for the Bermudan dynamic programming and for the conditional construction in Section~\ref{sec:conditional_pde}. Thus, the theorem is stated in the closed-domain form used in this paper. We do not study the reflected model, near-zero recurrence, strict positivity, or boundary inaccessibility.


\begin{assumption}[CEV exponents]\label{ass:delta-range}
Throughout this section we impose
$$
\delta_1,\delta_2\in[1/2,1].
$$
\end{assumption}

Let $x^+=\max\{x,0\}$. For $\delta\in[1/2,1]$, the inequality
$$
\bigl|(x^+)^\delta-(y^+)^\delta\bigr|
\le |x-y|^\delta,\qquad x,y\in\R,
$$
implies that the positive-part diffusion coefficient
$\sigma(x)=\xi(x^+)^\delta$ satisfies
$$
|\sigma(x)-\sigma(y)|
\le \xi |x-y|^\delta,\qquad x,y\in\R.
$$
This is the Yamada-Watanabe modulus condition~\cite{yamada1971uniqueness} in
the one-dimensional form needed below. More explicitly, the condition requires a
nondecreasing function $\varrho$ with $\varrho(0)=0$ and $\varrho(u)>0$ for
$u>0$ such that
$$
|\sigma(x)-\sigma(y)|\le \varrho(|x-y|),\qquad x,y\in\R,
$$
and
$$
\int_0^\varepsilon \frac{\dd u}{\varrho(u)^2}
=
\infty
$$
for sufficiently small $\varepsilon>0$. Here we may take
$\varrho(u)=\xi u^\delta$, since
$$
\int_0^\varepsilon \frac{\dd u}{\varrho(u)^2}
=
\xi^{-2}\int_0^\varepsilon u^{-2\delta}\,\dd u
=
\infty,\qquad \delta\ge\frac12.
$$
The upper bound $\delta\le1$ gives at most linear growth.

The state space used for the GDMR dynamics is
$$
\mathcal D=(0,\infty)\times[0,\infty)^2.
$$

\begin{theorem}[Strong well-posedness and nonnegativity]\label{thm:wellposed}
Assume Assumption~\ref{ass:delta-range}. For every $T>0$ and every initial
state
$$
S_0>0,\qquad v_0\ge0,\qquad v'_0\ge0,
$$
the nonnegative-state system \eqref{eq:gdmr} admits a unique strong solution
$(S_t,v_t,v'_t)_{t\in[0,T]}$ with continuous paths such that
$$
S_t>0,\qquad v_t\ge0,\qquad v'_t\ge0,\qquad
0\le t\le T,\quad \Q\text{-a.s.}
$$
Moreover, for every deterministic initial state in $\mathcal D$, the solution
is a time-homogeneous strong Markov process on $\mathcal D$.
\end{theorem}

Here, the boundary cases $v_0=0$ and $v'_0=0$ are included only to show that the variance system is well defined on the closed nonnegative domain and cannot leave it. They are not used as a modelling choice for calibration or numerical pricing. In the numerical part of the paper, the initial variance factors are strictly positive.

The proof is given in Appendix~\ref{app:proof-wellposed}. It follows the triangular structure of the model.
First, the $v'$-equation is solved, then the $v$-equation is solved with $v'$ as
an adapted input, and finally, the asset equation is recovered from its
stochastic exponential. The same one-way structure is used later in
Subsection~\ref{sec:hybrid-overview} and Proposition~\ref{prop:conditional-gaussian-step}, where the Hybrid LSMC--PDE step conditions on the Brownian increments driving the variance factors.

The lower bound for CEV exponents, $\delta_i\ge1/2$,  is the Yamada-Watanabe threshold for
pathwise uniqueness of power diffusion coefficients. We do not impose a
Feller-type condition, since the theorem proves nonnegativity rather than strict
positivity or boundary inaccessibility.

\section{Bermudan option pricing}\label{sec:conditional_pde}

We now pass from the model dynamics to the one-step continuation calculation in
the Bermudan recursion. The purpose of this section is to explain how the Hybrid LSMC--PDE method is used for the GDMR model and why the Brownian projection is needed.
The Hybrid LSMC--PDE method was developed by Farahany, Jackson, and Jaimungal~\cite{farahany2020mixing} for multidimensional stochastic volatility models. In this paper, our contribution is to adapt that general framework to the GDMR model by deriving the Brownian projection and the resulting one-dimensional conditional log-price representation. {As a by-product, the conditional representation also yields a semi-analytical pricing formula for European options under the GDMR model.} We use the
general hybrid-regression convergence theorem in \cite{farahany2020mixing}; the representation derived here provides the
model-specific input needed to apply that theorem.

\subsection{Bermudan stopping problem}\label{sec:bermudan-recursion}

Let
$$
\pi=\{0=t_0<t_1<\cdots<t_M=T\}
$$
be the Bermudan exercise grid. At $t_n$, the payoff is $h_n(S_{t_n})$, where
$h_n:(0,\infty)\to\R$ is Borel. We assume
$$
\E\!\left[
\max_{0\le n\le M}e^{-rt_n}|h_n(S_{t_n})|
\right]<\infty .
$$

We focus on Bermudan put options, with payoff $h_n(s)=(K-s)^+$, for which the early exercise feature is generally nontrivial. {
The Hybrid LSMC--PDE method is not inherently restricted
to put options and can, in principle, be extended to other
payoff structures.}

By Theorem~\ref{thm:wellposed}, the state process
$(S_t,v_t,v'_t)$ is a time-homogeneous Markov process on
$$
\mathcal D=(0,\infty)\times[0,\infty)^2 .
$$
Thus the Bermudan value functions may be defined by backward induction. Write
$$
\Delta t_n:=t_{n+1}-t_n
$$
and set the value function at maturity as
$$
U_M(s,v,w):=h_M(s),
$$
where $w$ denotes the second variance coordinate. For $n=M-1,\ldots,0$, define
$$
\begin{aligned}
C_n(s,v,w)
:=
e^{-r\Delta t_n}
\E\!\left[
U_{n+1}(S_{t_{n+1}},v_{t_{n+1}},v'_{t_{n+1}})
\,\middle|\,
S_{t_n}=s,\ v_{t_n}=v,\ v'_{t_n}=w
\right],
\end{aligned}
$$
and following the dynamic programming principle, the value function at time $t_n$ prior to maturity is
$$
U_n(s,v,w):=\max\{h_n(s),C_n(s,v,w)\}.
$$
The time-zero Bermudan value is
$$
V_0=U_0(S_0,v_0,v'_0).
$$

\subsection{Hybrid decomposition of the continuation value}\label{sec:hybrid-overview}

A plain LSMC method approximates the continuation value in the Bermudan recursion
by regressing simulated discounted next-step values on the current state. In a
stochastic-volatility model this regression is more demanding than in a one-factor
asset model, because the continuation value depends on both the asset level and the
variance state. The Hybrid LSMC--PDE method reduces this regression problem by
using the one-way coupling structure of the dynamics. The variance factors are simulated
first; along each simulated variance path, the remaining conditional expectation in
the asset direction is computed by a one-dimensional PDE step; the resulting
pathwise conditional values are then regressed over the current variance state. The
numerical realization of these steps, including the FFT evaluation and the regression
matrices, is given in Section~\ref{sec:hybrid-method}.

The idea can be written directly at the level of one Bermudan interval. Let $F$ be
a next-step value function; in the backward recursion, $F$ is $U_{n+1}$ or its
numerical approximation. For a current state $(s,v,w)$, define
$$
C_n^F(s,v,w)
=
e^{-r\Delta t_n}
\E\!\left[
F(S_{t_{n+1}},v_{t_{n+1}},v'_{t_{n+1}})
\,\middle|\,
S_{t_n}=s,\ v_{t_n}=v,\ v'_{t_n}=w
\right].
$$
The hybrid method evaluates this conditional expectation in two stages, using
the one-way coupling of the model. The equations for $(v,v')$ do not involve
$S$, so over one exercise interval we may first simulate the Brownian increments
driving the variance factors. Conditioning on these increments fixes the
variance segment and the part of the asset Brownian motion correlated with the
variance drivers. The only remaining asset randomness is the orthogonal
Brownian component. This conditioning is used only inside the one-step
conditional expectation and does not change the information used for the
exercise decision at the exercise dates.

Let $\mathcal H_n$ denote this enlarged conditioning information; its precise
definition is given in the next subsection. The inner, pathwise continuation value is

\begin{equation}\label{eq:innercontval}
    \bar C_n^F(s;\mathcal H_n):= e^{-r\Delta t_n} \E\!\left[ F(S_{t_{n+1}},v_{t_{n+1}},v'_{t_{n+1}})
\,\middle|\, S_{t_n}=s,\mathcal H_n \right].
\end{equation}

For each simulated variance path segment, $\bar C_n^F(\cdot;\mathcal H_n)$ is a
function of the initial asset level. The original continuation value is recovered by the outer conditional expectation
$$
C_n^F(s,v,w) = \E\!\left[ \bar C_n^F(s;\mathcal H_n) \,\middle|\, v_{t_n}=v,\ v'_{t_n}=w
\right].
$$
This expectation is the part approximated by least-squares regression over the current variance state.

To compute the continuation value, we must isolate the part of the asset Brownian motion that remains random after conditioning on $\mathcal H_n$. This is the role of the Brownian projection below. Once this residual component is identified, the conditional asset-price problem becomes one-dimensional in the log-price variable. The conditioning is only an iterated-expectation device and does not change the exercise information at time $t_n$.

\subsection{Brownian projection}

We now make the conditioning information in the preceding decomposition explicit.
Fix an interval $[t_n,t_{n+1}]$ in the Bermudan exercise grid and define
$$
\mathcal H_n := \F_{t_n}
\vee \sigma\!\left( W_s^{(2)}-W_{t_n}^{(2)},
W_s^{(3)}-W_{t_n}^{(3)}: s\in[t_n,t_{n+1}]
\right), 
$$
where $\F_1 \vee \F_2 := \sigma(\F_1 \cup \F_2)$.

Since $v$ and $v'$ are driven only by $W^{(2)}$ and $W^{(3)}$, the path
$(v_s,v'_s)_{s\in[t_n,t_{n+1}]}$ is $\mathcal H_n$-measurable. By assumption $|\rho_{23}|\neq1$, we may project
$W^{(1)}$ onto the span of $(W^{(2)},W^{(3)})$. Set
$$
\Sigma_{23}:=
\begin{bmatrix}
1 & \rho_{23}\\
\rho_{23} & 1
\end{bmatrix},
\qquad
\beta:=\Sigma_{23}^{-1}
\begin{bmatrix}
\rho_{12}\\
\rho_{13}
\end{bmatrix}
=
\begin{bmatrix}
\beta_2\\
\beta_3
\end{bmatrix}.
$$
Equivalently,
\begin{equation}\label{eq:proj-coeffs}
\beta_2=\frac{\rho_{12}-\rho_{13}\rho_{23}}{1-\rho_{23}^2},
\qquad
\beta_3=\frac{\rho_{13}-\rho_{12}\rho_{23}}{1-\rho_{23}^2}.
\end{equation}
The residual variance is
\begin{equation}\label{eq:sigperp}
\begin{aligned}
\sigma_\perp^2
&:=
1-
\begin{bmatrix}\rho_{12}&\rho_{13}\end{bmatrix}
\Sigma_{23}^{-1}
\begin{bmatrix}\rho_{12}\\ \rho_{13}\end{bmatrix} \\
&=
\frac{
1-\rho_{12}^2-\rho_{13}^2-\rho_{23}^2
+2\rho_{12}\rho_{13}\rho_{23}
}{1-\rho_{23}^2}
\ge0 .
\end{aligned}
\end{equation}
When $\sigma_\perp>0$,
$$
B_t^\perp
:=
\frac{W_t^{(1)}-\beta_2W_t^{(2)}-\beta_3W_t^{(3)}}{\sigma_\perp}
$$
is a Brownian motion, and its increments over $[t_n,t_{n+1}]$ are independent
of $\mathcal H_n$. Hence
\begin{equation}\label{eq:orth}
\dd W_t^{(1)}
=\beta_2\,\dd W_t^{(2)}
+\beta_3\,\dd W_t^{(3)}
+\sigma_\perp\,\dd B_t^\perp.
\end{equation}
If $\sigma_\perp=0$, the residual term is absent and the asset Brownian driver
is fully determined by the Brownian drivers of the variance factors.

\subsection{Conditional one-dimensional equation}

Let $X_t:=\log S_t$. Substituting Eq.~\eqref{eq:orth} into the asset price equation for $S_t$ in Eq.~\eqref{eq:gdmr} gives
$$
\dd X_t
=\Bigl(r-\tfrac12v_t\Bigr)\dd t
+\sqrt{v_t}\Bigl(\beta_2\,\dd W_t^{(2)}
+\beta_3\,\dd W_t^{(3)}\Bigr)
+\sigma_\perp\sqrt{v_t}\,\dd B_t^\perp .
$$
Define the $\mathcal H_n$-measurable shift
\begin{equation}\label{eq:Zn}
Z_n:=
\int_{t_n}^{t_{n+1}}
\sqrt{v_s}\Bigl(\beta_2\,\dd W_s^{(2)}+
\beta_3\,\dd W_s^{(3)}\Bigr).
\end{equation}
For $t\in[t_n,t_{n+1}]$, define also
\begin{equation}\label{eq:In-Bn}
I_n(t):=\int_t^{t_{n+1}}\Bigl(r-\tfrac12v_s\Bigr)\dd s,
\qquad
B_n(t):=\sigma_\perp^2\int_t^{t_{n+1}}v_s\,\dd s,
\end{equation}
and write $I_n:=I_n(t_n)$ and $B_n:=B_n(t_n)$.
\begin{proposition}[Conditional one-step representation]\label{prop:conditional-gaussian-step}
Let $G:\R\times[0,\infty)^2\to\R$ be bounded and Borel. Conditionally on
$\mathcal H_n$, the function
$$
u(t,y)
:=
\E\!\left[
G\!\left(
Y_{t_{n+1}}^{t,y}+Z_n,
v_{t_{n+1}},
v'_{t_{n+1}}
\right)
\,\middle|\,
\mathcal H_n
\right],
$$
where
$$
\dd Y_s^{t,y}
=\Bigl(r-\tfrac12v_s\Bigr)\dd s
+\sigma_\perp\sqrt{v_s}\,\dd B_s^\perp,
\qquad
Y_t^{t,y}=y,
$$
is the bounded mild solution of
$$
\partial_tu+
\Bigl(r-\tfrac12v_t\Bigr)\partial_yu+
\frac12\sigma_\perp^2v_t\partial_{yy}u=0,
\qquad
u(t_{n+1},y)=G(y+Z_n,v_{t_{n+1}},v'_{t_{n+1}}).
$$
Equivalently, for an auxiliary $\xi\sim N(0,1)$,
\begin{equation}\label{eq:gaussian_rep}
u(t,y)
=
\E_\xi\!\left[
G\!\left(
y+Z_n+I_n(t)+\sqrt{B_n(t)}\,\xi,
v_{t_{n+1}},
v'_{t_{n+1}}
\right)
\right],
\end{equation}
with the deterministic interpretation when $B_n(t)=0$.
\end{proposition}

\begin{proof}
After conditioning on $\mathcal H_n$, the variance path and the shift $Z_n$ are fixed. If $\sigma_\perp>0$, the only remaining randomness is the residual Brownian motion; if $\sigma_\perp=0$, the conditional step is deterministic. The conditional Feynman-Kac formula gives the backward Kolmogorov equation. Since the coefficients are independent of $y$, the terminal log-price is Gaussian with mean $y+I_n(t)$ and variance $B_n(t)$, which gives the displayed formula.
\end{proof}
For an asset-space next-step value function $F$ and $G(y,v,w):=F(e^y,v,w)$, the pathwise quantity defined by Eq.~\eqref{eq:innercontval} in Subsection~\ref{sec:hybrid-overview} is
$$
\bar C_n^F(s;\mathcal H_n)=e^{-r\Delta t_n}u(t_n,\log s).
$$
Thus the conditional asset step depends on the sampled $(W^{(2)},W^{(3)})$-path over $[t_n,t_{n+1}]$ only through
$$
Z_n,\qquad I_n,\qquad B_n,\qquad v_{t_{n+1}},\qquad v'_{t_{n+1}},
$$
while $(v_{t_n},v'_{t_n})$ is the state used in the outer regression.

{
To the best of our knowledge, no general closed-form pricing formula is available for European options under the non-affine GDMR model. However, Proposition~\ref{prop:conditional-gaussian-step} provides a useful semi-analytical representation for European put option prices. Indeed, by taking the single interval $[0,T]$, the conditional expectation in Eq.~\eqref{eq:gaussian_rep} can be evaluated analytically. Taking the remaining expectation and discounting then yields the European put price, which can be estimated by Monte Carlo simulation. The result is stated in the following corollary.
}
{
\begin{corollary}[European option pricing under the GDMR model]
\label{cor:european}
Consider a European put with strike $K>0$ and maturity $T$. Taking $t_0=0$ and $t_1=T$, let $Z_0$, $I_0$, and $B_0$ be defined by Eqs.~\eqref{eq:Zn} and \eqref{eq:In-Bn}, and set 
\[
m_0=\log S_0+Z_0+I_0,\qquad
d_1=\frac{m_0-\log K+B_0}{\sqrt{B_0}},
\qquad d_2=d_1-\sqrt{B_0}.
\]
Then
\begin{equation}\label{eq:european-put}
P_0=e^{-rT}\E\left[
K\Phi(-d_2)-e^{m_0+B_0/2}\Phi(-d_1)
\right],
\end{equation}
where $\Phi$ denotes the standard normal distribution function and the expectation is taken over the randomness driving the variance processes. When $B_0=0$, the expression inside the expectation is interpreted as $(K-e^{m_0})^+$. The European call price follows from put--call parity, $C_0=P_0+S_0-Ke^{-rT}$.
\end{corollary}
\begin{proof}
Taking $t_n=0$, $t_{n+1}=T$, and $G(y,v,w)=(K-e^y)^+$ in Proposition~\ref{prop:conditional-gaussian-step} gives \[
u(0,\log S_0)
=\E_\xi\left[(K-e^{m_0+\sqrt{B_0}\xi})^+\right].
\]
For $B_0>0$, evaluating the Gaussian integral yields 
\[
u(0,\log S_0)
=K\Phi(-d_2)-e^{m_0+B_0/2}\Phi(-d_1).
\]
Taking the expectation over the variance-driving Brownian motions and discounting by $e^{-rT}$ yields Eq.~\eqref{eq:european-put}.  The call formula follows from the European put--call parity.
\end{proof}
Eq.~\eqref{eq:european-put} is a conditional Monte Carlo formula that is straightforward to implement. The outer expectation is approximated by simulating \(N\) independent realizations of the variance and stochastic long-run mean processes, driven by their underlying Brownian motions.  For each realization, we compute $Z_0$ according to Eq.~\eqref{eq:Zn} and the integrated variance $A_T=\int_0^T v_s\,\dd s$, from which $I_0=rT-\frac12 A_T$ and $B_0=\sigma_\perp^2 A_T$ by Eq.~\eqref{eq:In-Bn}. The conditional European put values are then evaluated analytically for each realization and averaged to obtain the conditional Monte-Carlo estimator for the time-zero option price.
}

{
The implementation of Bermudan option pricing is more involved than that of European option pricing, as the backward induction additionally requires the treatment of early exercise and regression-based approximation of continuation values. Proposition~\ref{prop:conditional-gaussian-step} provides the conditional propagation step within the backward induction, while the full numerical implementation is described in Section~\ref{sec:hybrid-method}.
}

\subsection{Connection with the hybrid convergence theorem}\label{sec:as_convergence}

{The almost-sure convergence of the regression coefficients in the idealized Hybrid LSMC--PDE framework is established by Farahany, Jackson, and Jaimungal~\cite{farahany2020mixing} under suitable assumptions. The preceding subsections identify the model-specific
ingredients relevant to that theorem.}

First, the asset equation is multiplicative. For $0\le t<u\le T$,
$$
S_u=S_tR_{t,u},
$$
where
$$
R_{t,u}
=
\exp\left(
\int_t^u\left(r-\frac12v_s\right)\,\dd s
+
\int_t^u\sqrt{v_s}\,\dd W_s^{(1)}
\right).
$$
Thus the return factor $R_{t,u}$ depends on the Brownian and variance paths,
but not on the value of $S_t$.

Second, Proposition~\ref{prop:conditional-gaussian-step} shows that, on each
Bermudan interval, the sampled $(W^{(2)},W^{(3)})$-path enters the conditional
asset-price step only through the finite-dimensional statistic
$$
\Theta_n
:=
(V_n,V_{n+1},Z_n,I_n,B_n),
\qquad
V_n=(v_{t_n},v'_{t_n}).
$$
Because the variance equations do not involve $S$, the model is one-way
coupled. {These structural properties are consistent with the framework considered in \cite[Theorem~1]{farahany2020mixing}.}

{
The cited convergence theorem also imposes technical conditions
on the regression basis, the sample Gram matrices, and the payoff
functions. The Bermudan put payoff is bounded, although it does
not satisfy the compact-support assumption without truncation.
As discussed in \cite{farahany2020mixing}, this assumption can
be accommodated by a suitable modification of the payoff near
the boundary of the asset-price domain. 
Motivated by these conditions, the numerical implementation
in Section~\ref{sec:hybrid-method} uses a compact volatility
domain $\mathcal D_v$, bounded regression basis functions, and
a stabilized sample Gram inverse $[A_n^N]^{-1}_R$.
The conditional PDE is solved numerically on a finite log-price
grid $\mathcal S_h$, corresponding to a truncated asset-price interval.}

{
Under the assumptions of \cite[Theorem~1]{farahany2020mixing},
the regression coefficients of the idealized hybrid method
converge almost surely as the number of simulated paths increases.
This provides a theoretical foundation for the hybrid regression
framework, while practical implementations necessarily involve
additional numerical approximations. In particular, the present implementation involves additional approximations, including time discretization of the variance processes, finite-grid FFT propagation, and interpolation, which are not covered by the cited theorem.
The accuracy and numerical performance of the implemented method
are investigated in Section~\ref{sec:numerical} and \ref{sec:heston-validation}.}

\subsection{Rationale for the Hybrid LSMC--PDE method}
\label{sec:rationale}

The main motivation for the Hybrid LSMC--PDE method is to reduce the extent of functional approximation required in plain LSMC. In plain LSMC, the continuation value is approximated by regression on prescribed basis functions of all state variables $(S,v,v')$. In the hybrid method, the dependence on $S$ is instead treated through the conditional one-dimensional pricing PDE derived from the model dynamics, while regression is restricted to the variance factors $(v,v')$. This reduces the dimensionality of the regression approximation without requiring an analytical solution for the option price.

Another useful feature is that the hybrid method computes the option value over an asset-price grid in a single backward induction. This allows the important option price sensitivity measure, Delta,  to be obtained by numerical differentiation with negligible additional
cost. The resulting price curve can also be used to obtain prices for different strikes through a homogeneity relation. These features, together with computational comparisons against plain LSMC, are
examined in Section~\ref{sec:computational_efficiency}.

Finally, treating the asset-price dependence through a conditional PDE provides flexibility for contracts with asset-price-dependent features. For example, discretely monitored knock-out conditions can be imposed directly on the asset-price value function at each monitoring date, as explained in
Section~\ref{sec:extensions}.

\section{Numerical implementation}
\label{sec:hybrid-method}

The numerical experiments for Bermudan option pricing use the following implementation. The conditional
representation is derived in Section~\ref{sec:conditional_pde}; here we specify the grids, variance discretization, accumulated path statistics, FFT step, regression step, and time-zero estimator. 
For the Bermudan options considered below, we use FFT to evaluate the conditional one-step representation in Proposition~\ref{prop:conditional-gaussian-step}, as its Gaussian structure allows the shift operation to be handled efficiently in Fourier space.


\subsection{Grids and basis functions}
Let
$$
\pi=\{0=t_0<t_1<\cdots<t_M=T\}
$$
be the Bermudan exercise grid, and let $h_n$ denote the exercise payoff at
$t_n$. We work on a uniform log-price grid
$$
y_i=y_{\min}+(i-1)\Delta y,
\qquad
\Delta y=\frac{y_{\max}-y_{\min}}{N_S-1},
\qquad
i=1,\ldots,N_S,
$$
and set
$$
s_i=e^{y_i},
\qquad
\mathcal S_h:=\{s_1,\ldots,s_{N_S}\}.
$$
The interval $[y_{\min},y_{\max}]$ is chosen large enough to contain the
relevant asset values used in the computation.

For the volatility variables, write $w=v'$. We work on the compact volatility rectangle
\begin{equation}\label{eq:vol-domain}
\mathcal D_v=[0,v_{\max}]\times[0,w_{\max}]\subset[0,\infty)^2 .
\end{equation}
The regression basis is chosen as a vector of continuous functions that are bounded on this rectangle,
\begin{equation}\label{eq:vol-basis}
\phi=(\phi_1,\ldots,\phi_{d_B}), \qquad \sup_{(v,w)\in\mathcal D_v}|\phi_k(v,w)|<\infty,
\qquad k=1,\ldots,d_B.
\end{equation}
For the convergence statement, the numerical basis may be viewed as the restriction to $\mathcal D_v$ of compactly supported continuous functions on the full volatility state space; this extension does not affect the computed regression, since the regression is evaluated only on $\mathcal D_v$.

Once chosen, the domain $\mathcal D_v$ and the basis functions are fixed before
the training regression.

At each exercise date, we approximate the continuation value by a separable
regression form
$$
c_n(s,v,w)\approx a_n(s)\cdot\phi(v,w),
\qquad
a_n(s)\in\R^{d_B}.
$$
Numerically, the coefficient vectors $a_n(s)$ are computed only at the grid points
$s_i\in\mathcal S_h$. Off-grid asset values are evaluated by interpolation in the
log-price coordinate.

\subsection{Variance simulation and path statistics}
Let
$$
0=\tau_0<\tau_1<\cdots<\tau_L=T
$$
be an Euler simulation grid containing all exercise dates in the Bermudan exercise grid $\pi$. Write
$$
\Delta\tau_\ell:=\tau_{\ell+1}-\tau_\ell.
$$
For each path ${j=1,\ldots,N}$, simulate correlated Gaussian increments
$$
(\Delta W_{\ell}^{(2),j},\Delta W_{\ell}^{(3),j})
$$
with variances $\Delta\tau_\ell$ and covariance
$\rho_{23}\Delta\tau_\ell$. To simulate the variance processes, we adopt a standard positivity-preserving truncated Euler scheme of the form:
$$
\begin{aligned}
\bar w_\ell^j&:=(w_\ell^j)^+,
\qquad
\bar v_\ell^j:=(v_\ell^j)^+,\\
w_{\ell+1}^j
&:=
\left[
w_\ell^j+\kappa_2(\theta-\bar w_\ell^j)\Delta\tau_\ell
+\xi_2(\bar w_\ell^j)^{\delta_2}\Delta W_{\ell}^{(3),j}
\right]^+,\\
v_{\ell+1}^j
&:=
\left[
v_\ell^j+\kappa_1(\bar w_\ell^j-\bar v_\ell^j)\Delta\tau_\ell
+\xi_1(\bar v_\ell^j)^{\delta_1}\Delta W_{\ell}^{(2),j}
\right]^+.
\end{aligned}
$$
For each exercise interval $[t_n,t_{n+1}]$, the path statistics are accumulated
over the Euler substeps contained in that interval:
$$
\begin{aligned}
Z_n^j
&:=
\sum_{\tau_\ell\in[t_n,t_{n+1})}
\sqrt{\bar v_\ell^j}
\left(
\beta_2\Delta W_{\ell}^{(2),j}
+
\beta_3\Delta W_{\ell}^{(3),j}
\right),\\
I_n^j
&:=
\sum_{\tau_\ell\in[t_n,t_{n+1})}
\left(r-\frac12\bar v_\ell^j\right)\Delta\tau_\ell,\\
B_n^j
&:=
\sigma_\perp^2
\sum_{\tau_\ell\in[t_n,t_{n+1})}
\bar v_\ell^j\Delta\tau_\ell .
\end{aligned}
$$
Here the projection coefficients $\beta_2,\beta_3$ and the residual volatility coefficient $\sigma_\perp$ are given by Eqs.~\eqref{eq:proj-coeffs} and~\eqref{eq:sigperp}, respectively. The left-point rule is natural for $Z_n^j$, since $Z_n^j$ approximates a stochastic integral evaluated with left endpoints. The same left-point values are used for $I_n^j$ and $B_n^j$ for consistency with the Euler  simulation, and the same increments $\Delta W_{\ell}^{(2),j},\Delta W_{\ell}^{(3),j}$ are used in the variance updates and in $Z_n^j$.

\subsection{Pathwise FFT propagation and regression}
Fix $n\in\{0,\ldots,M-1\}$. Here $\widehat V_{n+1}^N$ denotes the numerical
Bermudan value surface already available from the next exercise date in the
backward recursion. At maturity this surface is initialized by the payoff function
$$
\widehat V_M^N(s_i,v,w)=h_M(s_i).
$$
At earlier exercise dates it is the previously updated surface
$$
\begin{aligned}
\widehat V_{n+1}^N(s_i,v,w)
&=
\max\{h_{n+1}(s_i),\widehat C_{n+1}^N(s_i,v,w)\},\\
\widehat C_{n+1}^N(s_i,v,w)
&=
a_{n+1}^N(s_i)\cdot\phi(v,w).
\end{aligned}
$$
Thus $\widehat V_{n+1}^N$ is already available before the following pathwise
propagation step. For each Monte Carlo path $j$, the simulated Brownian
increments driving the variance factors on $[t_n,t_{n+1}]$ determine the
endpoint state
$$
(v_{t_{n+1}}^j,w_{t_{n+1}}^j)
$$
and the statistics $Z_n^j,I_n^j,B_n^j$. Define the terminal grid function
$$
g_{n+1}^j(y_i)
:=
\widehat V_{n+1}^N
\left(e^{y_i},v_{t_{n+1}}^j,w_{t_{n+1}}^j\right),
\qquad i=1,\ldots,N_S.
$$
By the Gaussian representation Eq.~\eqref{eq:gaussian_rep} in Section~\ref{sec:conditional_pde},
$$
u_n^j(y)
=
\E_\xi\!\left[
g_{n+1}^j
\left(
y+Z_n^j+I_n^j+\sqrt{B_n^j}\xi
\right)
\right],
\qquad \xi\sim N(0,1).
$$
With the Fourier convention
$$
\mathcal F f(\omega)
=
\int_{\R}e^{-\mathrm i\omega y}f(y)\,\dd y,
$$
we obtain
$$
\mathcal F u_n^j(\omega)
=
\mathcal F g_{n+1}^j(\omega)
\exp\left(
\mathrm i\omega Z_n^j+\mathrm i\omega I_n^j-\frac12\omega^2B_n^j
\right).
$$
Define
$$
\Psi_n^j(\omega)
:=
\mathrm i\omega Z_n^j+\mathrm i\omega I_n^j-\frac12\omega^2B_n^j.
$$
On the uniform log-price grid, the FFT implementation is
\begin{equation}
u_n^j(\cdot)
\approx
\operatorname{FFT}^{-1}
\left(
\operatorname{FFT}[g_{n+1}^j]\,
\exp(\Psi_n^j)
\right),
\label{eq:fft-recursion}
\end{equation}
where the exponential is applied pointwise at the discrete Fourier frequencies.
The FFT input is the unshifted terminal grid $g_{n+1}^j$; the correlated shift
$Z_n^j$ is included only through the multiplier. Equivalently, one may shift the
terminal grid first and omit $\mathrm i\omega Z_n^j$, but the two conventions
must not be combined. If $B_n^j=0$, the conditional propagation reduces to the
deterministic log-price shift $Z_n^j+I_n^j$.

The pathwise pre-surface on the asset grid is
$$
\widehat C_n^j(s_i)
:=
e^{-r\Delta t_n}u_n^j(y_i),
\qquad i=1,\ldots,N_S.
$$

The regression step is performed at each exercise date $t_n$,
$n=M-1,\ldots,1$, and at each asset grid point $s_i$. It estimates the
dependence of the pathwise pre-surface $\widehat C_n^j(s_i)$ on the current
volatility state. Define
$$
V_n^j:=(v_{t_n}^j,w_{t_n}^j).
$$
The sample Gram matrix is
$$
A_n^N
:=
\frac1N\sum_{j=1}^N
\phi(V_n^j)\phi(V_n^j)^\top,
$$
and the right-hand side for the grid point $s_i$ is
$$
b_n^N(s_i)
:=
\frac1N\sum_{j=1}^N
\phi(V_n^j)\widehat C_n^j(s_i).
$$
For numerical stability, we use the truncated inverse
\begin{equation}\label{eq:estimated_coeff}
[A]^{-1}_R
:=
\begin{cases}
A^{-1}, & \text{if $A$ is invertible and }\|A^{-1}\|\le R,\\
0,  & \text{otherwise.}
\end{cases}
\end{equation}
Then set
\begin{equation}
a_n^N(s_i):=[A_n^N]^{-1}_R b_n^N(s_i).
\label{eq:lsq}
\end{equation}
The fitted continuation surface is
$$
\widehat C_n^N(s_i,v,w)
:=
a_n^N(s_i)\cdot\phi(v,w),
$$
and the updated price for the Bermudan option is
$$
\widehat V_n^N(s_i,v,w)
:=
\max\left\{
h_n(s_i),
\widehat C_n^N(s_i,v,w)
\right\}.
$$
Starting from
$$
\widehat V_M^N(s_i,v,w):=h_M(s_i),
$$
the recursion is performed backward for $n=M-1,\ldots,1$. The fitted functions
are defined at the asset grid points and extended to off-grid asset values by
log-price interpolation.

\subsection{Time-zero estimator and pseudocode}
At $t_0=0$, the initial volatility state is fixed, so no regression is needed. After $\widehat V_1^N$ has been fitted, the time-zero
continuation value is evaluated on the asset grid by averaging first-step
pathwise pre-surfaces:
$$
\widehat C_0^N(s_i)
:=
\frac1N
\sum_{j=1}^{N}
\widehat C_0^j(s_i).
$$
The averaged grid values define $\widehat C_0^N$ on the asset grid; off-grid
values are evaluated by the log-price interpolation convention stated above.
Thus $\widehat C_0^N(S_0)$ denotes the averaged continuation value at the
initial asset price, and the Hybrid LSMC--PDE estimator is
\begin{equation}\label{eq:hybrid_est}
\widehat V_0^N(S_0)
=
\max\left\{
h_0(S_0),
\widehat C_0^N(S_0)
\right\}.
\end{equation}

The implementation of the Hybrid LSMC--PDE method is summarized in Algorithm~\ref{alg:hybrid-LSMC--PDE}.
\begin{algorithm}[H]
\caption{Hybrid LSMC--PDE recursion}\label{alg:hybrid-LSMC--PDE}
\begin{algorithmic}[1]
\Require Exercise grid $\{t_n\}_{n=0}^M$, Euler grid containing the exercise dates, log-price grid $\{y_i\}_{i=1}^{N_S}$, volatility domain $\mathcal D_v$, basis $\phi$, truncation level $R$, and path count $N$
\Ensure Hybrid LSMC--PDE estimator $\widehat V_0^N(S_0)$ and fitted continuation surfaces $\widehat C_n^N$
\State Simulate $N$ training paths of the variance factors on the Euler grid.
\State Store exercise-date volatility states and accumulate $Z_n^j$, $I_n^j$, and $B_n^j$ over each interval $[t_n,t_{n+1}]$.
\State Set $\widehat V_M^N(s_i,v,w)\gets h_M(s_i)$ on $\mathcal S_h\times\mathcal D_v$.
\For{$n=M-1,\ldots,1$}
\For{$j=1,\ldots,N$}
\State Compute $\{\widehat C_n^j(s_i)\}_{i=1}^{N_S}$ using the unshifted terminal grid $g_{n+1}^j$ in Eq.~\eqref{eq:fft-recursion}, with $Z_n^j$ included only through the multiplier $\exp(\Psi_n^j)$; do not pre-shift $g_{n+1}^j$ by $Z_n^j$.
\EndFor
\For{each $s_i\in\mathcal S_h$}
\State Compute $b_n^N(s_i)$ and $a_n^N(s_i)$ by the truncated-inverse regression Eq.~\eqref{eq:lsq}.
\EndFor
\State Set $\widehat C_n^N(s_i,v,w)\gets a_n^N(s_i)\cdot\phi(v,w)$ for all $s_i\in\mathcal S_h$.
\State Set $\widehat V_n^N(s_i,v,w)\gets\max\{h_n(s_i),\widehat C_n^N(s_i,v,w)\}$ for all $s_i\in\mathcal S_h$.
\EndFor
\State Compute the first-step pre-surfaces $\widehat C_0^j(s_i)$ for $j=1,\ldots,N$ using the simulated paths on $[t_0,t_1]$.
\State Set $\widehat C_0^N(s_i)\gets N^{-1}\sum_{j=1}^{N}\widehat C_0^j(s_i)$.
\State Evaluate $\widehat C_0^N(S_0)$ from the averaged continuation grid at $y=\log S_0$, using linear interpolation if needed.
\State \textbf{return} $\widehat V_0^N(S_0)=\max\{h_0(S_0),\widehat C_0^N(S_0)\}$.
\end{algorithmic}
\end{algorithm}

We now examine the computational complexity of the Hybrid LSMC--PDE method. Recall that $N$, $N_S$, $M$, $L$, and $d_B$ denote the number of simulated variance paths, asset-price grid points, the number of Bermudan exercise dates, Euler simulation steps, and regression basis functions,
respectively. The simulation of the variance paths and the accumulation of the path statistics require $O(NL)$ operations.
The conditional FFT propagation in Eq.~\eqref{eq:fft-recursion}
is performed for all $N$ paths over the $M$ Bermudan time intervals, requiring $O(MNN_S\log N_S)$ operations.  The least-squares regression is performed at the $M-1$ intermediate exercise dates. However, For a fixed number $d_B$ of regression basis functions, the computational cost of the least-squares regressions across the asset-price grid is $O(MNN_S)$. 
This does not change the overall asymptotic complexity, which is dominated by
the conditional FFT propagation. Therefore, the overall computational complexity of the Hybrid LSMC--PDE method is
\[
O\bigl(NL+MNN_S\log N_S\bigr).
\]

\subsection{Extensions to knock-out options}
\label{sec:extensions}
The grid-based implementation in the Hybrid LSMC--PDE method can also accommodate asset-price-dependent contract features imposed at discrete monitoring dates. As an illustration, consider an up-and-out Bermudan put option with an upper barrier $H$, where the barrier monitoring dates coincide with the Bermudan exercise dates. Between consecutive monitoring dates, the conditional continuation value is propagated using the same FFT procedure as described previously. At each monitoring date,
the Bermudan exercise decision is applied over the surviving region
$y\in[y_{\min},\log H)$ (equivalently, $S<H$), while the option value
is set to zero for $S\geq H$. Thus, the barrier feature can be
incorporated directly and conveniently on the asset-price grid without changing the conditional solver or the simulation of the variance processes.

For a discretely monitored knock-out option of European style, neither least-squares regression nor early exercise checks are needed. The conditional value is propagated backward through the barrier-monitoring dates using the same FFT procedure, with the knock-out condition imposed at each monitoring date. The time-zero option price is then estimated by averaging the resulting conditional values over all simulated variance paths.

\section{Numerical experiments {under the GDMR model}}\label{sec:numerical}

For numerical illustration, we consider Bermudan put pricing experiments using the fixed parameter set in Table~\ref{tab:parameter-setup}. Since no closed-form solution is available for Bermudan option prices under the GDMR model, {and no established reference prices are available in the literature, we use a high-budget plain LSMC implementation to obtain numerical reference prices. This provides a basis for assessing the accuracy of the proposed method under the same model and parameter settings.} We compare the plain LSMC method~\cite{longstaff2001valuing} and the Hybrid LSMC--PDE estimator in Eq.~\eqref{eq:hybrid_est} against {these reference prices}. Although they are not exact solutions, they serve as numerical {references for the comparisons in this section}. {For simplicity, we refer to these relative deviations from the reference prices as relative errors.} 

\subsection{Experimental setting and reference values}

Table~\ref{tab:parameter-setup} reports the fixed GDMR parameter set and numerical settings used in the experiments throughout Section~\ref{sec:numerical}.
\begin{table}[htbp]
\centering
\caption{Parameter set and numerical settings used in the pricing experiments.}
\label{tab:parameter-setup}
\begin{threeparttable}
\begin{tabularx}{0.94\textwidth}{@{}lr>{\raggedright\arraybackslash}X@{}}
\toprule
Symbol & Value & Interpretation \\
\midrule
$S_0$ & 100 & initial asset price \\
$v_0$ & 0.114 & initial instantaneous variance factor \\
$v'_0$ & 0.110 & initial secondary variance factor \\
$r$ & 0.02 & risk-free rate \\
$T$ & 1.0 & maturity in years \\
$\kappa_1$ & 5.5 & mean-reversion speed of $v_t$ toward $v'_t$ \\
$\kappa_2$ & 0.1 & mean-reversion speed of $v'_t$ toward $\theta$ \\
$\theta$ & 0.078 & long-run variance level \\
$\xi_1$ & 2.689 & volatility scale of the first variance factor \\
$\xi_2$ & 0.502 & volatility scale of the second variance factor \\
$\delta_1$ & 0.94 & CEV exponent of the first variance factor \\
$\delta_2$ & 0.94 & CEV exponent of the second variance factor \\
$\rho_{12}$ & -0.982 & correlation between the first and second Brownian drivers \\
$\rho_{13}$ & -0.727 & correlation between the first and third Brownian drivers \\
$\rho_{23}$ & 0.590 & correlation between the second and third Brownian drivers \\
\addlinespace
\multicolumn{3}{@{}l}{\textit{Fixed numerical setting}} \\
$N_{\mathrm{ex}}$ & 12 & exercise dates after $t_0$, including maturity \\
$K$ & $\{70,80,90,100,110\}$ & strike prices used in all experiments \\
$N_S$ & 301 & asset-grid points in the conditional PDE solver \\
$(y_{\min},y_{\max})$ & $(\log (0.30S_0),\log(3.5\max \left(S_0, K\right)))$ & lower and upper log-price grid bounds \\
\botrule
\end{tabularx}
\end{threeparttable}
\end{table}

A volatility truncation quantile $q_v=0.999$ is used to choose the bounds
$v_{\max}$ and $w_{\max}$ in the compact volatility domain
$\mathcal D_v$ in Eq.~\eqref{eq:vol-domain}. Thus, $v_{\max}$ and $w_{\max}$ are chosen so that approximately $99.9\%$ of
the simulated values of $v$ and $v'$ at the exercise dates lie below the corresponding bounds, truncating only the upper $0.1\%$ tail. The Hybrid LSMC--PDE grid and truncation settings are kept fixed throughout all experiments, unless stated otherwise. 

{The mean-reversion, volatility-scale, initial-variance, and correlation parameters, together with the CEV exponents $\delta_1,\delta_2$,  are taken from the calibrated double mean-reverting specification of Bayer, Gatheral and Karlsmark~\citep{bayer2013fast}. Specifically, we use their parameter set calibrated to SPX and VIX option market data from September 15, 2011 (Table~2 in their paper). While Bayer et al.~\citep{bayer2013fast} assume $r=0$, we set $r=0.02$ to reflect more realistic market conditions,  noting that the interest rate is typically determined independently of the stochastic volatility parameters. Thus, our numerical experiments use a published market-calibrated parameter set as a fixed test case.} The correlation matrix corresponding to the parameters in Table~\ref{tab:parameter-setup} satisfies the assumed positive-definiteness condition, since $1-\rho_{12}^2\approx 0.0357>0$ and $\det(\rho)\approx 0.0015 >0$. The restriction on the leverage correlations $\rho_{12},\rho_{13}\in(-1,0]$ is also satisfied.

The plain LSMC regressions use a heuristically chosen 16-term basis, combining polynomial terms in the state variables with powers of the normalized put payoff:
$$
1,x,y,z,x^2,y^2,z^2,xy,xz,yz,x^3,y^3,z^3,p,p^2,p^3,
$$
where $x=S/S_0$, $y=v/\theta$, $z=v'/\theta$, and
$p=(K-S)^+/K$. {Since the hybrid method handles the asset-price dimension through the conditional PDE, its regressions are restricted to the volatility factors, using a complete cubic polynomial basis:} 
$$
1,\tilde{y},\tilde{z},\tilde{y}^2,\tilde{y}\tilde{z},\tilde{z}^2,\tilde{y}^3,\tilde{y}^2\tilde{z},\tilde{y}\tilde{z}^2,\tilde{z}^3,
$$
where
$
\tilde{y}=\min(v,v_{\max})/v_{\max} \text{ and }
\tilde{z}=\min(v',v'_{\max})/v'_{\max}.$


The reference prices, denoted by $V^{\mathrm{ref}}$ and reported in Table~\ref{tab:reference-prices}, are high-budget plain LSMC estimates computed using $1200$ Euler steps and $1{,}200{,}000$ simulation paths. The table also reports empirical Monte Carlo standard errors and associated 95\% confidence intervals for completeness, although the reference prices $V^{\mathrm{ref}}$ are treated as fixed values in subsequent relative-error calculations and plots, without accounting for their estimation uncertainty.

For a time-zero option price estimator $\widehat V$ and reference price $V^{\mathrm{ref}}$, we report the relative error
\begin{equation}\label{eq:relative-error}
\varepsilon_{\mathrm{rel}}
:=
\left| \frac{\widehat V-V^{\mathrm{ref}}}{V^{\mathrm{ref}}} \right|.
\end{equation}

{To quantify the statistical uncertainty of the estimators by the hybrid and the plain LSMC methods, we compute the standard error $\mathrm{SE}$ and construct approximate 95\% confidence intervals $[V_{\mathrm{left}},V_{\mathrm{right}}]=[\widehat V-1.96\,\mathrm{SE},\,\widehat V+1.96\,\mathrm{SE}]$ for both methods.
In the relative-error plots in Section~\ref{sec:varying Euler} and \ref{sec:varying path}, the error bars are obtained by calculating the minimum and maximum relative errors over the corresponding price confidence interval:
\[
\varepsilon_{\mathrm{lower}}
=\min_{V\in[V_{\mathrm{left}},V_{\mathrm{right}}]}
\frac{|V-V^{\mathrm{ref}}|}{|V^{\mathrm{ref}}|},
\qquad
\varepsilon_{\mathrm{upper}}
=\max_{V\in[V_{\mathrm{left}},V_{\mathrm{right}}]}
\frac{|V-V^{\mathrm{ref}}|}{|V^{\mathrm{ref}}|}.
\]
In particular, $\varepsilon_{\mathrm{lower}}=0$ if the price confidence interval contains $V^{\mathrm{ref}}$. }

\begin{table}[htbp]
\centering
\caption{LSMC reference values for the parameter set in Table~\ref{tab:parameter-setup}.}
\label{tab:reference-prices}
\begin{tabular}{@{}rrrr@{}}
\toprule
$K$ & Reference price ($V^{\mathrm{ref}}$) & Std. error & 95\% confidence interval \\
\midrule
70  & 2.522851  & 0.007103 & [2.508929, 2.536773] \\
80  & 4.298933  & 0.009345 & [4.280617, 4.317249] \\
90  & 6.942119  & 0.011774 & [6.919042, 6.965196] \\
100 & 10.682037 & 0.014105 & [10.654391, 10.709683] \\
110 & 15.743520 & 0.015757 & [15.712636, 15.774404] \\
\botrule
\end{tabular}
\end{table}

\subsection{Fixed path budget and varying Euler steps}
\label{sec:varying Euler}
The first experiment keeps the Monte Carlo path budget fixed and varies the number of
Euler time steps. The two path budgets are $20{,}000$ and $60{,}000$, and
both methods are evaluated at $24,48,72,$ and $96$ Euler steps. This
experiment tests the sensitivity of the estimators to the Euler
discretization.

For $20{,}000$ paths, the relative errors of the plain LSMC estimator, as defined in Eq.~\eqref{eq:relative-error}, range from $0.344\%$ to $8.645\%$ over the full step-sweep grid, while those of the Hybrid LSMC--PDE estimator range from 
$0.019\%$ to $7.459\%$. Figure~\ref{fig:step-sweep-20k} reports the strike-wise
relative errors. 


\begin{figure}[H]
\centering
\includegraphics[width=\textwidth]{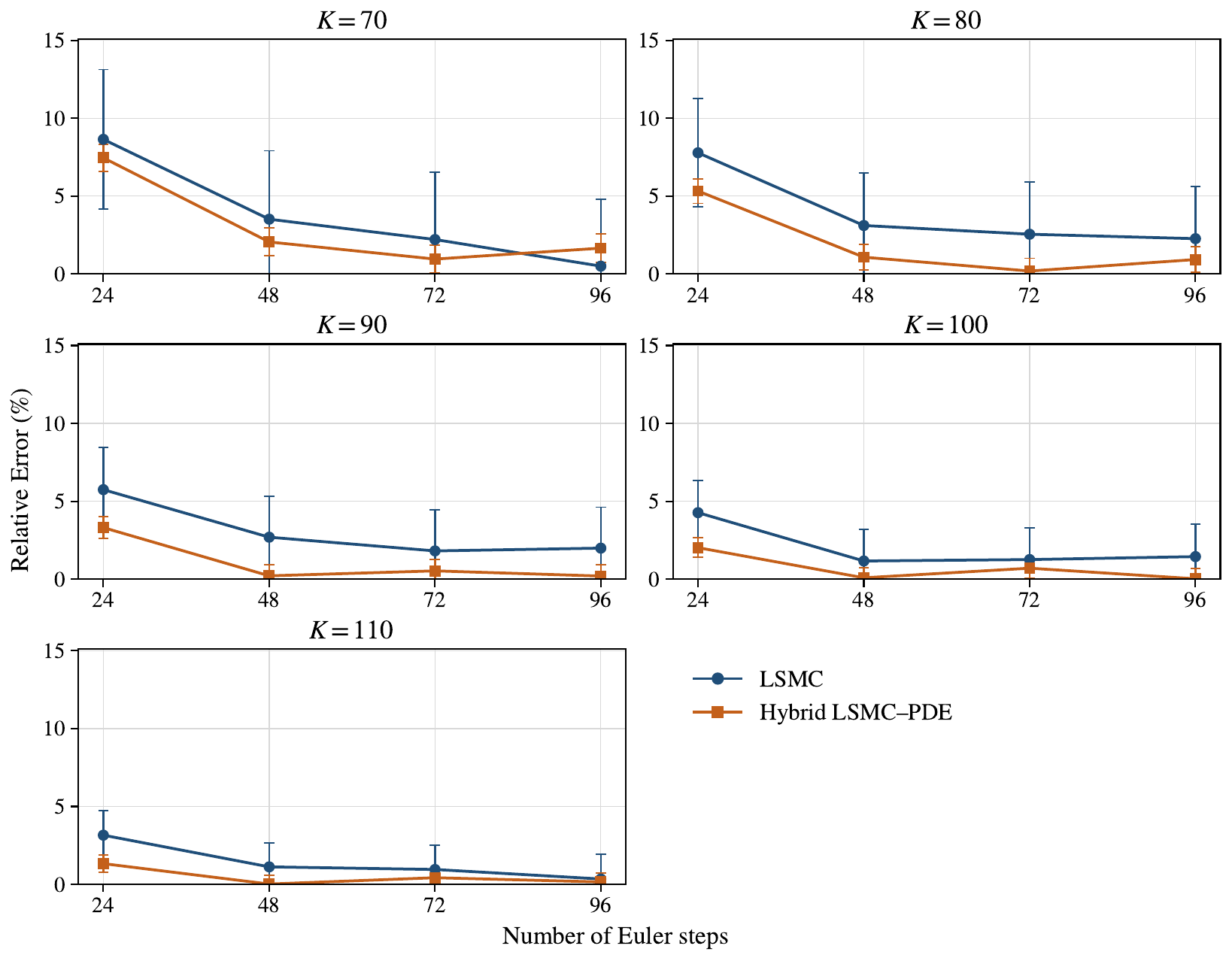}
\caption{Relative errors with associated error bars for  $20{,}000$-path step-sweep experiment.}
\label{fig:step-sweep-20k}
\end{figure}

For $60{,}000$ paths, the same design gives relative errors of the plain LSMC estimator ranging from $0.102\%$ to $7.720\%$, while those of the Hybrid LSMC--PDE estimator range from $0.028\%$ to $7.857\%$. Figure~\ref{fig:step-sweep-60k} reports the corresponding strike-wise relative errors. The Hybrid LSMC--PDE estimator does not give a
smaller relative error in every step-count configuration, but it gives smaller
relative errors in most of the reported configurations, with the clearest
differences at the intermediate step levels.

\begin{figure}[H]
\centering
\includegraphics[width=\textwidth]{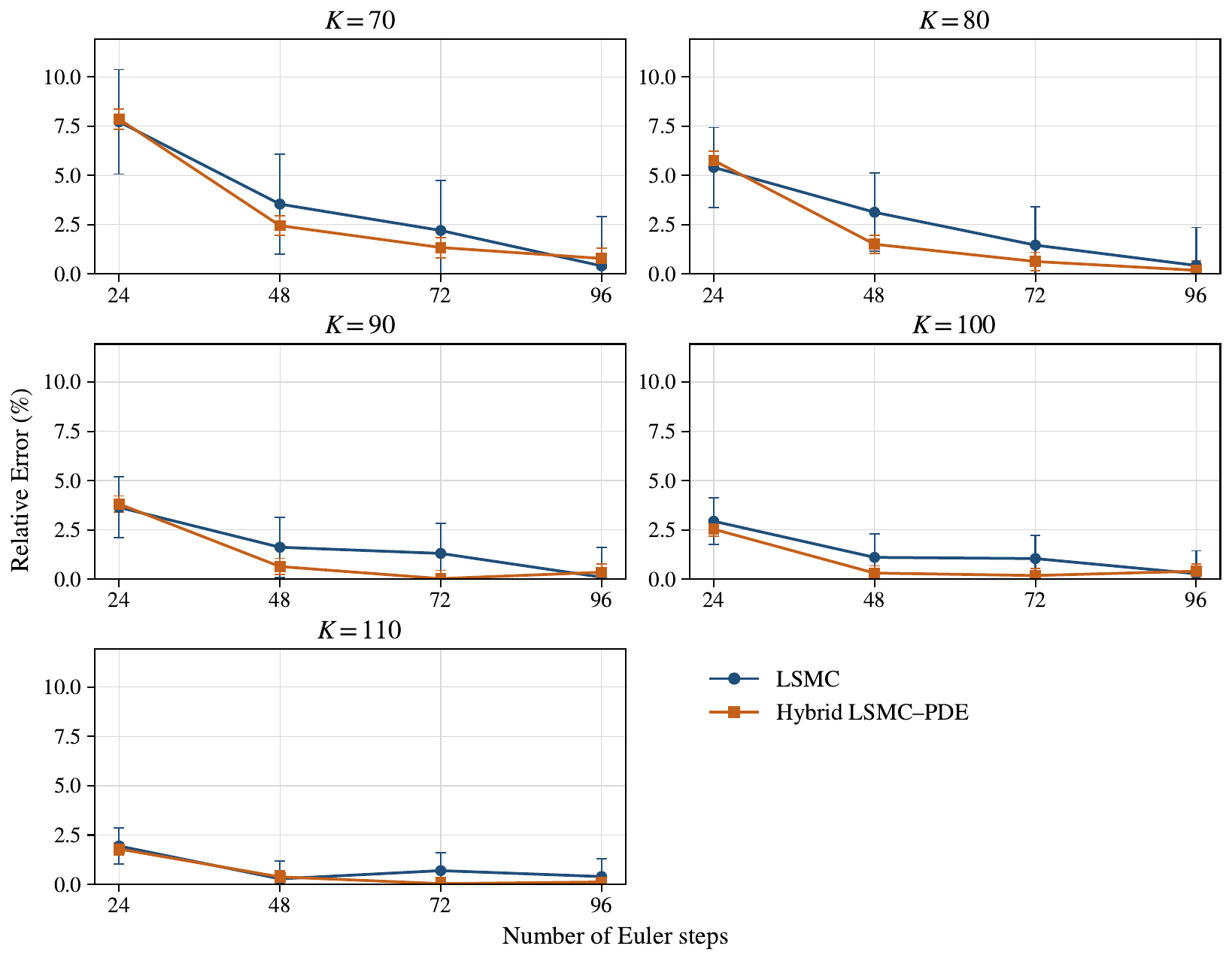}
\caption{Relative errors with associated error bars for $60{,}000$-path step-sweep experiment.}
\label{fig:step-sweep-60k}
\end{figure}

\FloatBarrier

\subsection{Fixed Euler step count and varying path budget}
\label{sec:varying path}
The second experiment fixes the number of Euler time steps and varies the Monte Carlo
path budget. The reported path counts are $250,1000,5000,10000,20000,40000$,
and $60000$. Two discretization levels, $48$ and $60$ Euler steps, are
considered. This experiment focuses on sampling variation and regression error.

At $48$ Euler steps, the at-the-money strike $K=100$ gives relative errors of the plain LSMC estimator ranging from $1.105\%$ to $25.948\%$ over the path grid, while those of the Hybrid LSMC--PDE estimator range from $0.088\%$ to $11.118\%$. Figure~\ref{fig:path-sweep-48} shows the corresponding path-count sweep for all five strikes. 

\begin{figure}[H]
\centering
\includegraphics[width=\textwidth]{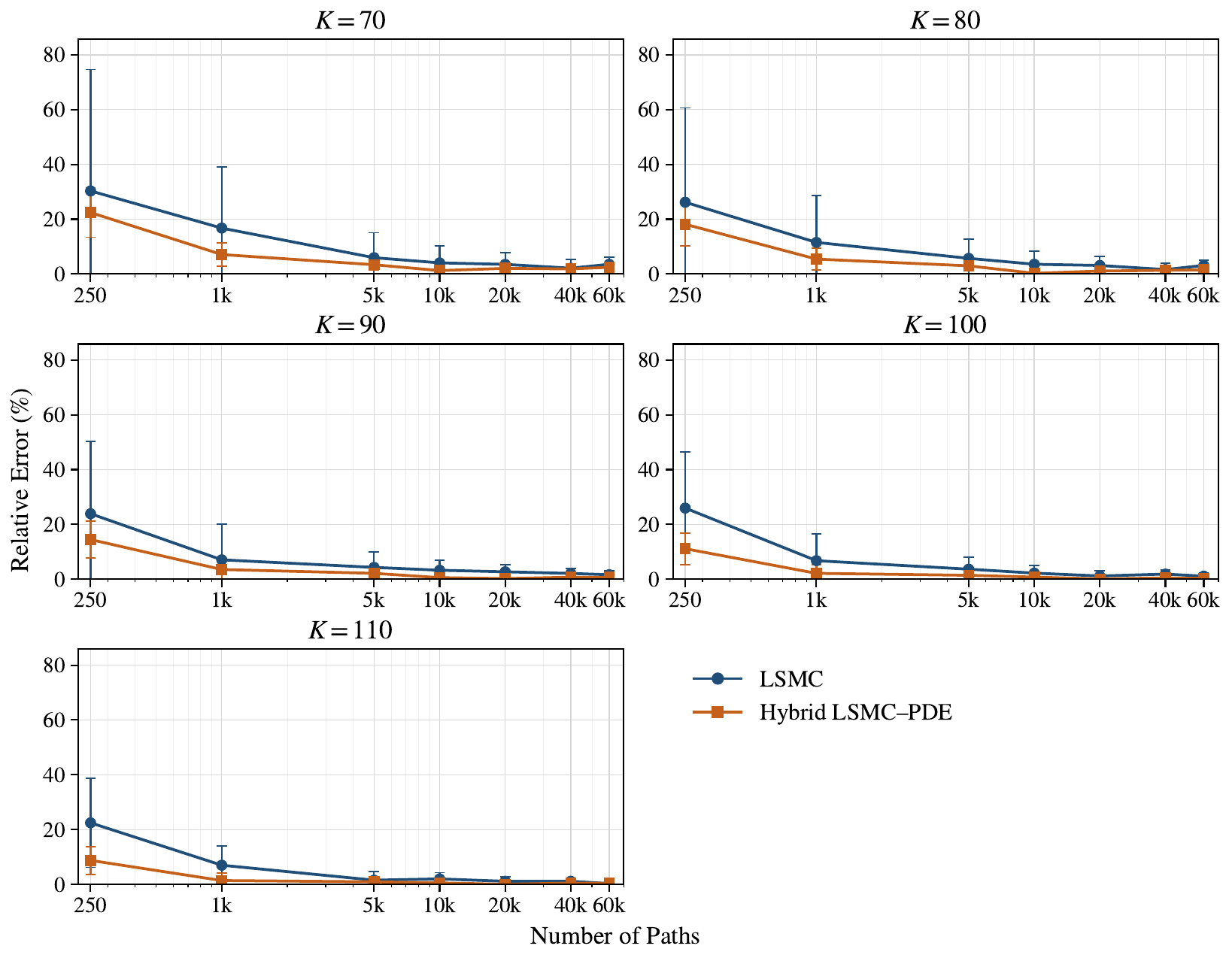}
\caption{Relative errors with associated error bars for the fixed $48$-step path-sweep experiment.}
\label{fig:path-sweep-48}
\end{figure}

At $60$ Euler steps, the corresponding relative errors of the plain LSMC estimator range from $1.324\%$ to $28.932\%$, while those of the Hybrid LSMC--PDE estimator range from $0.049\%$ to $10.102\%$. Figure~\ref{fig:path-sweep-60} reports the corresponding path-count sweep for all five strikes. The reduction in relative error is most visible in the low- and moderate-path regimes, where the plain LSMC estimator is most affected by sampling noise.

\begin{figure}[H]
\centering
\includegraphics[width=\textwidth]{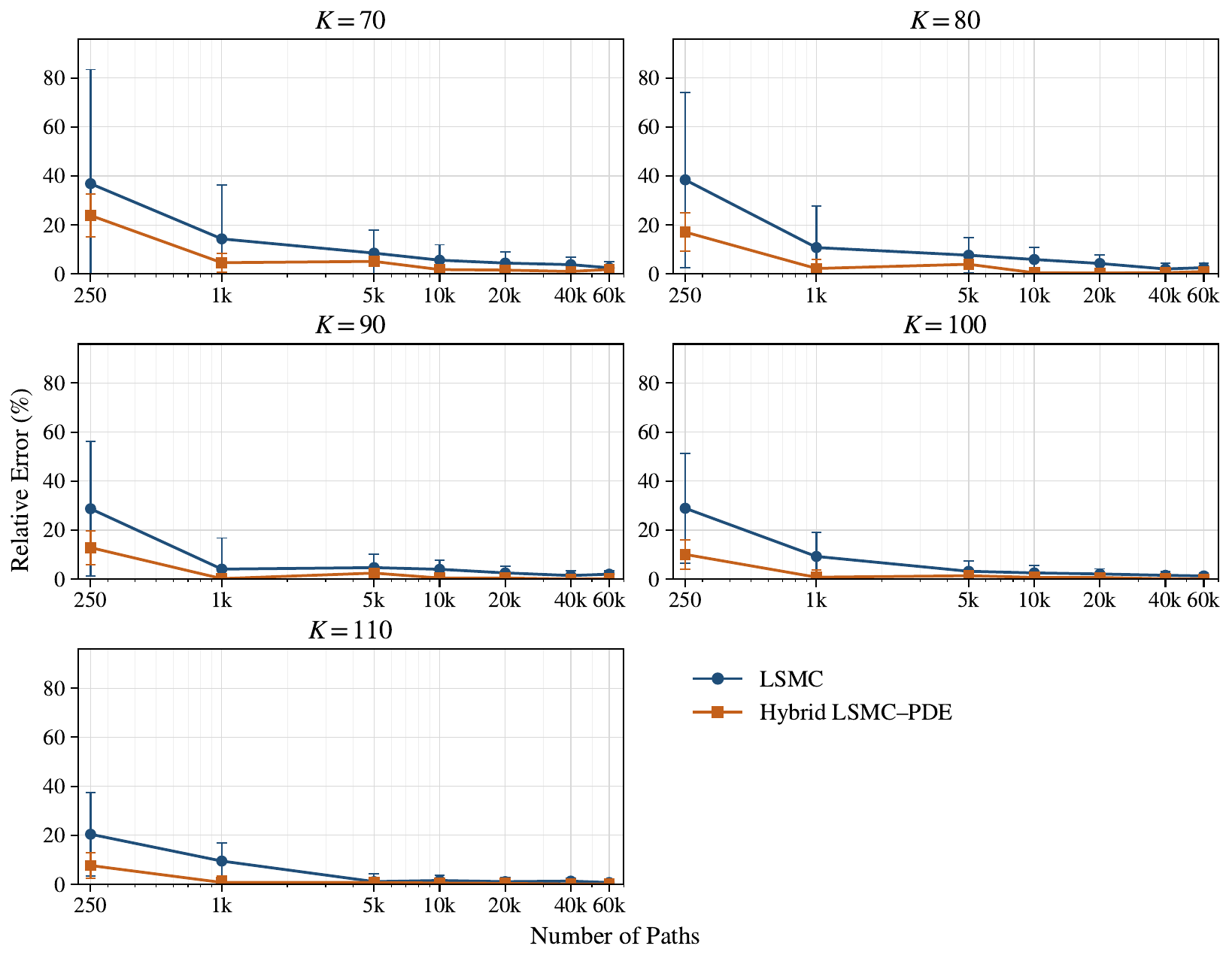}
\caption{Relative errors with associated error bars for the fixed $60$-step path-sweep experiment.}
\label{fig:path-sweep-60}
\end{figure}

\FloatBarrier
\subsection{Discussion of pricing accuracy and uncertainty}
\label{sec:discussionGMDR}
As shown in Figures~\ref{fig:step-sweep-20k}--\ref{fig:path-sweep-60}, the error bars are generally largest in the low-path experiments, especially for the plain LSMC estimator. In comparison, the Hybrid LSMC--PDE estimator is usually gives a smaller observed relative error and {has generally narrower error bars} in the low and moderate path regimes. As the number of paths increases, the error bars shrink for both methods, and the difference between plain LSMC and Hybrid LSMC--PDE becomes smaller.

In general, the experimental results are consistent with the rationale discussed in Section~\ref{sec:rationale}, namely that the conditional PDE step reduces the dimensionality of the continuation regression. It is worth reiterating that the reference prices reported in Table~\ref{tab:reference-prices} are not exact option values. Indeed, these reference prices are obtained using the same methodology as plain LSMC, but with a substantially larger computational budget. Nevertheless, the Hybrid LSMC--PDE method achieves competitive precision despite using a different numerical approach, which is encouraging and may be attributed to the conditional PDE treatment of the stock-price dependence.

The complete set of estimated option prices is reported in Appendix~\ref{app:price-grids} and may serve as {reference values} for future studies. More specifically, the pricing estimates underlying Figures~\ref{fig:step-sweep-20k}--\ref{fig:path-sweep-60} are reported in Tables~\ref{tab:app-step-sweep-prices-20k}--\ref{tab:app-path-sweep-60-prices}, respectively.

\subsection{Computational efficiency and stock-price dependence}
\label{sec:computational_efficiency}

The preceding numerical experiments demonstrate that the Hybrid LSMC--PDE method generally achieves competitive pricing accuracy compared with the plain LSMC under the same simulation settings. We now compare the computational costs of the two methods under identical simulation settings. Throughout this subsection, we use 60,000 simulated paths, 96 Euler steps, and 12 exercise dates. As in the previous experiemnts, we use $N_S=301$ spatial grid intervals for the hybrid method unless stated otherwise. All computations in this subsection were performed using Python 3.11.7 with NumPy 1.26.4 on an HP EliteBook 860 G10 laptop equipped with a 13th-generation Intel Core i5-1335U processor and 16 GB of RAM.

The hybrid method is designed to compute the option value function over a grid of initial stock prices, with values at intermediate stock prices obtained by interpolation. However, the preceding experiments considered only the option at one single initial stock price $S_0$. We now compare the computational costs of the two methods for pricing both single and multiple options. We begin with a single Bermudan put with $S_0=K=100$. Over 20 independent runs, the average computation times were 0.665 seconds for plain LSMC and 41.492 seconds for the hybrid method, with standard deviations of 0.074 and 4.550 seconds, respectively. Thus, when only a single option price is required, plain LSMC is substantially more computationally efficient under the present configuration.

Recall that the GDMR model is one-way coupled: the volatility dynamics are independent of the stock price. Consequently, scaling the initial stock price by a positive constant scales the entire stock-price process by the same constant, while leaving the volatility dynamics unchanged.  Since the put payoff scales proportionally when both the stock price and strike are scaled, the Bermudan put value satisfies the following homogeneity property
\[
V(cS_0,cK)=cV(S_0,K), \qquad c>0.
\]
In particular, the normalized option price $V(S_0,K)/K$ depends only on the moneyness $S_0/K$, when all other model parameters are fixed. For a  reference initial stock price $S_0$ and strike $K_0$, the homogeneity property above gives
\begin{equation}
\label{eq:homogeneity_relations}
V(S_0,K)
=\frac{K}{K_0}
  V\left(\frac{S_0K_0}{K},K_0\right); \qquad V(S,K_0)
=\frac{K_0}{K}V(S_0,K),
\quad K=\frac{S_0K_0}{S}.
\end{equation}

The hybrid method directly computes the option price function over a grid of initial stock prices for a fixed strike $K_0$. By the first relation in Eq.~\eqref{eq:homogeneity_relations}, prices for different strikes can then be obtained by interpolation, without repeating the backward induction.

Plain LSMC can also efficiently price multiple strikes by reusing the same simulated paths, although the regression-based backward induction must be repeated for each strike. Thus, when only a small number of strike prices is required, plain LSMC is generally more computationally efficient.

Figure~\ref{fig:multiple_strike_timing} compares the total execution times for 10, 50, 100, 301, and 512 strikes (or equivalently, initial stock prices) over the same stock-price range. For plain LSMC, the  reported times are averaged over 20 independent runs for each number of evaluation points. For the hybrid method, the average computation time of 41.49 seconds is taken from the 20 independent runs of the single-strike experiment, with the spatial resolution fixed at $N_S=301$. Since the hybrid method computes the entire stock-price curve in a single backward induction, the additional cost of evaluating the option at multiple points by interpolation is negligible, even for 512 evaluation points. Therefore, the same computation time is used for all five strike counts. The computational time of plain LSMC increases approximately linearly with the number of evaluation points. For 100 points, plain LSMC requires 27.14 seconds on average, compared with 41.49 seconds for the hybrid method. However, for 301 points, the hybrid method becomes more efficient, requiring 41.49 seconds compared with 79.26 seconds for plain LSMC. For 512 points, the corresponding times are 41.49 and 135.23 seconds, respectively.
\begin{figure}[htbp]
    \centering
    \includegraphics[width=0.7\textwidth]{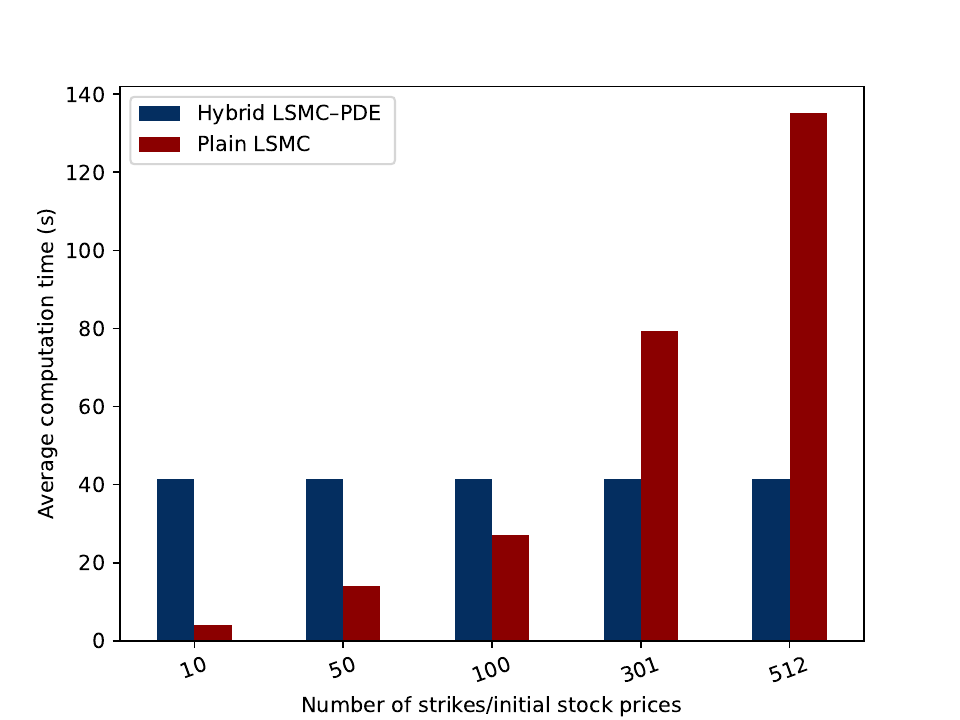}
    \caption{Comparison of total execution times for plain LSMC
    and Hybrid LSMC--PDE with 10, 50, 100, 301, and 512
    strikes (or initial stock prices). The bars represent
    the average computation times over 20 independent runs.}
    \label{fig:multiple_strike_timing}
\end{figure}

The timing results indicate that neither method is uniformly more computationally efficient. Plain LSMC is preferable when a limited number of option prices is required, whereas the hybrid method becomes competitive when a dense price curve is the object of interest.

Beyond computational efficiency, we also compare the numerical
regularity of the option price and Delta curves, which is often of practical interest.  The hybrid method computes the price curve directly on its spatial grid through a single backward induction, whereas plain LSMC requires separate regression-based backward inductions for different strikes.

To examine the numerical regularity of the curves at a fine resolution, we consider 512 equally spaced initial stock prices over $S_0\in[60,160]$, with the strike fixed at $K_0=100$. For plain LSMC, the price curve is reconstructed from 512 corresponding strike prices using the second relation in Eq.~\eqref{eq:homogeneity_relations}. For the hybrid method, we use $N_S=512$ spatial grid points and evaluate the resulting price curve at the same initial stock prices. In both cases, Delta is estimated by finite differences, without smoothing.

The computed option price curves by the two methods exhibit no substantial visual differences and are therefore omitted. However, the differences become more apparent in the corresponding Delta curves. As shown in  Figure~\ref{fig:delta_comparison}, both methods produce a similar overall shape for the Delta curves, but plain LSMC exhibits noticeable oscillations, whereas the hybrid method yields a much smoother Delta curve, illustrating a potential practical advantage. The oscillations observed in the plain LSMC Delta curve can be attributed to the amplification of numerical errors when finite differences are applied to separately estimated option prices, even though the same set of simulated paths are used for these estimates.  More specialized Monte Carlo techniques for estimating option price sensitivities, such as Delta, are discussed in \cite{glasserman2004montecarlo} with pathwise differentiation and likelihood-ratio methods as two main approachs.

\begin{figure}[htbp]
    \centering
    \includegraphics[width=0.49\textwidth]
        {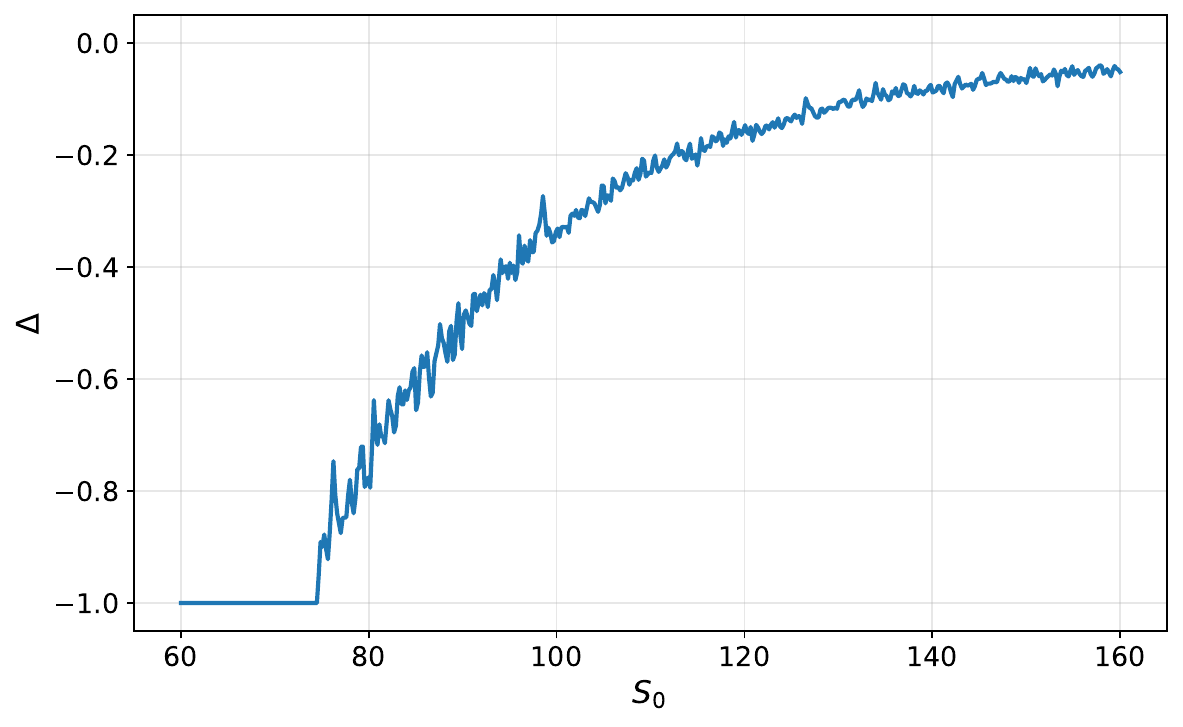}
    \hfill
    \includegraphics[width=0.49\textwidth]
        {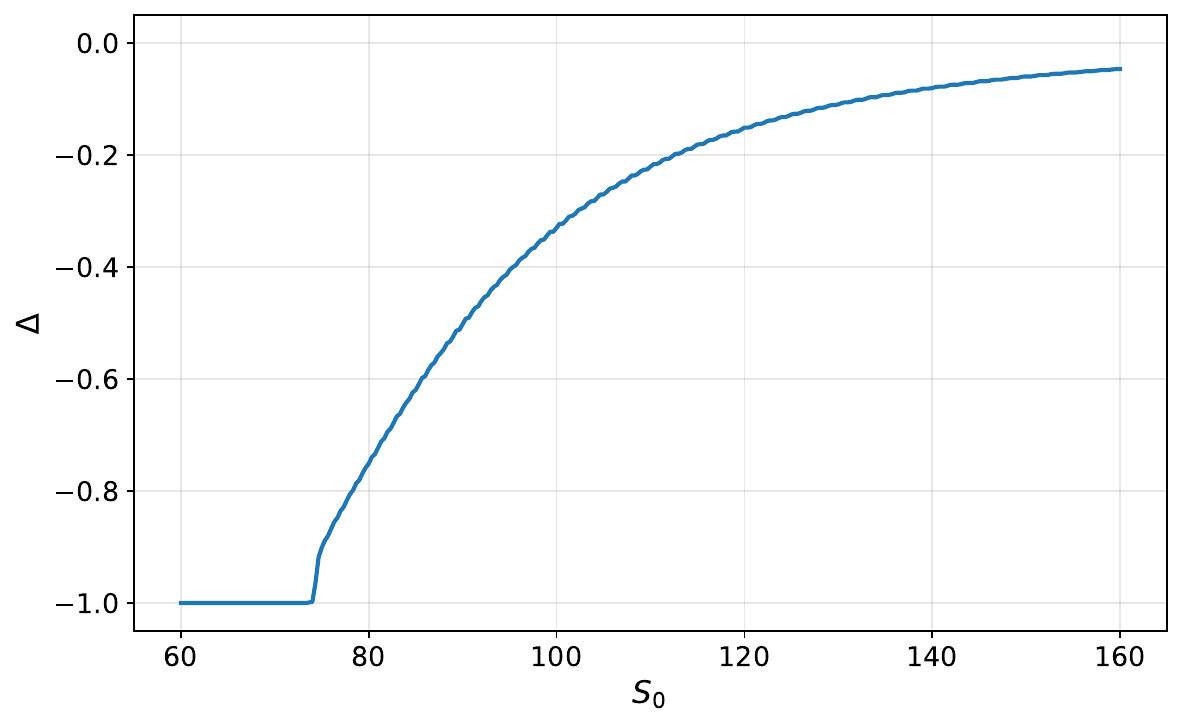}
    \caption{Finite-difference Delta curves for plain LSMC
with 512 strikes (left) and Hybrid LSMC--PDE with
$N_S=512$ (right). No smoothing is applied.}
    \label{fig:delta_comparison}
\end{figure}

\section{Numerical validation in the Heston setting}
\label{sec:heston-validation}

In Section~\ref{sec:numerical}, reference prices were obtained using a high-budget LSMC method, as alternative established pricing methods under the GDMR model are limited. These reference prices are themselves subject to statistical uncertainty and hence additional validation of the Hybrid LSMC--PDE method is desirable. 

We note that the GDMR model includes the Heston model ~\cite{heston1993closed} as a special case. Specifically, setting \(\xi_2=0\), \(v'_0=\theta\), and \(\delta_1=1/2\) gives \(v'_t=\theta\), reducing the GDMR variance dynamics to those of the Heston model.   

Under this specialization, more accurate reference prices can be obtained using established techniques that do not involve Monte-Carlo simulation, so we can further validate the hybrid method. In particular, European option prices can be computed using a semi-analytical Fourier-based formula, while Bermudan option prices can be obtained using the deterministic COS method.

In addition, experiments in the present section provide a complementary assessment of statistical uncertainty. Whereas Section~\ref{sec:numerical} presents relative-error bars derived from  approximate 95\% confidence intervals based on within-run Monte Carlo standard errors, the experiments in this section use independent runs to compute the root mean square error (RMSE).

The parameters $\delta_2$, $\kappa_2$, $\rho_{13}$ and $\rho_{23}$ do not affect the Heston dynamics and can be chosen freely subject to the admissibility conditions of the GDMR model. For both the European and Bermudan experiments, we set
$\delta_2=1/2$, $\kappa_2=1$ and  and $\rho_{13}=\rho_{23}=0$. The resulting $3\times3$ correlation matrix is positive definite whenever $|\rho_{12}|<1$, since its determinant is $1-\rho_{12}^2>0$. 

For a given numerical configuration, let $\widehat V_j(c)$ denote the estimated option price for contract $c$ in independent run $j$, and let $V^{\mathrm{ref}}(c)$ denote its reference price, computed independently using either a Fourier-based semi-analytical formula or the COS method. The RMSE for contract $c$ is defined by
\[
\mathrm{RMSE}(c)
=
\left(
\frac{1}{R}\sum_{j=1}^{R}
\bigl(\widehat V_j(c)-V^{\mathrm{ref}}(c)\bigr)^2
\right)^{1/2},
\]
where $R$ is the number of independent runs. We evaluate \(C\) contracts for each numerical configuration, where $C=15$ for European options and $C=9$ for Bermudan
options, and define the aggregate RMSE as
\[
\mathrm{RMSE}_{\mathrm{agg}}
=
\left(
\frac{1}{C}\sum_{c=1}^{C}
\mathrm{RMSE}(c)^2
\right)^{1/2}.
\]
Here, $c=(K,T)$ for European options and
$c=(S_0,M)$ for Bermudan options.

The European and Bermudan experiments use different simulation budgets and time-discretization settings to examine the methods under different numerical configurations. The European experiments vary both the number of simulation paths and the number of Euler steps per year, with 15 independent runs per configuration. For the Bermudan experiments, which also consider increasing numbers of exercise dates, we use a fixed total of 120 Euler steps over $[0,T]$ to ensure that each exercise interval contains multiple simulation steps.

Moreover, for each of the nine contracts,  15 independent runs are performed for European options in Section~\ref{sec:heston-european} and 10 for Bermudan options in Section~\ref{sec:heston-bermudan} due to the higher computational cost of Bermudan pricing. The resulting RMSE captures both bias and run-to-run variability of the whole pricing procedure.

\subsection{European option pricing: comparison with the Heston semi-analytical formula}
\label{sec:heston-european}

We first validate the European put pricing representation in Corollary~\ref{cor:european} under the Heston specialization of the GDMR model. We compare the conditional Monte Carlo (conditional MC) estimator in Eq.~\eqref{eq:european-put} with independent reference prices obtained from the well-known semi-analytical Heston call pricing formula~\cite{heston1993closed}, evaluated using the Fourier integral representation of Lewis~\cite{lewis2001simple}, followed by put--call parity. The Fourier pricing formula and its numerical implementation are detailed in Appendix~\ref{app:heston-fourier}.

We use the Heston parameter set reported in~\cite{CAO2021}, calibrated to S\&P 500 option market data from 2015:
\[
\begin{gathered}
v_0=0.16070^2,\qquad
\kappa_1=9.589,\qquad
\theta=0.16562^2,\\
\xi_1=0.332,\qquad
\rho_{12}=-0.246,
\end{gathered}
\]
with $S_0=100$,  $r=0.02$.  We consider European puts with strikes $K\in\{80,90,100,110,120\}$ and maturities $T\in\{1/12,1,2\}$, giving 15 strike--maturity combinations (contracts).

For each contract, we compare two estimators. The plain Monte Carlo (plain MC) estimator generates the terminal stock price directly as \[
S_T=\exp\left(m_0+\sqrt{B_0}\,\xi\right),
\qquad \xi\sim N(0,1),
\] where $m_0$ and $B_0$ are computed from the simulated
variance paths as defined in Corollary \ref{cor:european}. This avoids simulating the asset-price process step by step with an Euler scheme and provides a computationally efficient
implementation of plain MC.  The conditional MC estimator instead evaluates the conditional put price analytically using Eq.~\eqref{eq:european-put} for each simulated realization of the variance-driving Brownian motions, and then averages these conditional values. Both estimators use the same simulated variance paths.

We consider $N\in\{10\,000,20\,000,40\,000,80\,000\}$
simulation paths and 24, 48, or 96 Euler steps per year. Each configuration is repeated independently 15 times. For each numerical configuration, we report the aggregate RMSE defined above, using the Heston Fourier prices as reference values.

\begin{table}[htbp]
\centering
\caption{European put pricing under the Heston specialization
of the GDMR model. Aggregate RMSE values are computed against
the Heston semi-analytical reference prices over 15 contracts
and 15 independent runs.
The $\mathrm{RMSE}_{\mathrm{agg}}$ ratio is the plain MC $\mathrm{RMSE}_{\mathrm{agg}}$ divided by the
conditional MC $\mathrm{RMSE}_{\mathrm{agg}}$.}
\label{tab:heston-european}
\begin{tabular}{@{}rrrrr@{}}
\toprule
$N$ & Steps/year & Plain MC & Conditional MC & $\mathrm{RMSE}_{\mathrm{agg}}$ ratio \\
\midrule
$10{,}000$ & 24 & 0.07578 & 0.02713 & 2.8 \\
$20{,}000$ & 24 & 0.05846 & 0.01809 & 3.2 \\
$40{,}000$ & 24 & 0.04666 & 0.01298 & 3.6 \\
$80{,}000$ & 24 & 0.03736 & 0.01073 & 3.5 \\
\addlinespace[2pt]
$10{,}000$ & 48 & 0.10661 & 0.02314 & 4.6 \\
$20{,}000$ & 48 & 0.05312 & 0.01667 & 3.2 \\
$40{,}000$ & 48 & 0.04134 & 0.01255 & 3.3 \\
$80{,}000$ & 48 & 0.02594 & 0.00717 & 3.6 \\
\addlinespace[2pt]
$10{,}000$ & 96 & 0.07226 & 0.02322 & 3.1 \\
$20{,}000$ & 96 & 0.06457 & 0.01509 & 4.3 \\
$40{,}000$ & 96 & 0.04232 & 0.01278 & 3.3 \\
$80{,}000$ & 96 & 0.03450 & 0.00731 & 4.7 \\
\botrule
\end{tabular}
\end{table}

As shown in Table~\ref{tab:heston-european}, for a fixed number of steps per year, the aggregate RMSE values for both methods decrease as the number of simulation paths increases. However, the conditional MC estimator achieves lower aggregate RMSE than plain MC across all 12 configurations. The RMSE reduction factors range from approximately 2.8 to 4.7. The key insight is that the conditional estimator benefits from analytically integrating out the independent Brownian component of the asset-price dynamics, thereby reducing the sampling variability.

We also compared the computational costs of the two  estimators for $N=40{,}000$, $T=1$, and 48 time steps per year, pricing all five strikes and using the same hardware and software environment as in Section~\ref{sec:computational_efficiency}. For this configuration, the aggregate RMSE of plain MC was 3.3 times that of conditional MC, as given in Table~\ref{tab:heston-european}. Averaged over 20 independent runs, the total runtimes were 0.3312 seconds for plain MC and 0.3629 seconds for conditional MC. Thus, conditional MC achieved a substantial reduction in pricing error with only a 9.6\% increase in total computational time.  Additional timing experiments were conducted across all considered combinations of simulation budgets, maturities, and time discretizations, showing only modest increases in computational time for one- and two-year maturities (approximately 2--21\%). For the shortest maturity ($T=1/12$), the relative increase was substantially higher, however the absolute runtimes of both methods remained small, with both methods requiring at most approximately 0.2 seconds to price five strikes across all numerical configurations.

Detailed put prices and RMSE values for individual contracts under various simulation settings are reported in Appendix~\ref{app:heston-european}. Conditional MC achieved lower RMSE than plain MC in 178 out of a total of 180 configurations, indicating a consistent reduction in pricing error and providing independent numerical validation of the conditional MC formula for European option prices in Corollary~\ref{cor:european}.

\subsection{Bermudan option pricing: comparison with COS reference prices}
\label{sec:heston-bermudan}

We next validate the implementation of the Hybrid LSMC--PDE method under the Heston specialization of the GDMR model by comparing Bermudan put prices with those reported by Fang and Oosterlee~\cite{FangOosterlee2011} in Test~5, Table~5. Their Fourier cosine (COS) prices provide an external reference for both the Hybrid LSMC--PDE and plain LSMC estimators. We use the same parameter settings of the test:
\[
\begin{gathered}
v_0=0.0348,\qquad
\kappa_1=1.15,\qquad
\theta=0.0348,\\
\xi_1=0.39,\qquad
\rho_{12}=-0.64,
\end{gathered}
\]
with $r=0.04$. We fix $K=100, T=0.25$ and consider the nine combinations of $S_0\in\{90,100,110\}$ and $M\in\{20,40,60\}$, resulting in nice contracts. We allow exercise at $t_m=mT/M$, $m=1,\ldots,M$, and the option price is the continuation value at time zero.

For each contract, we perform $10$ independent runs of each method at $N=5{,}000$, $10{,}000$ and $20{,}000$ paths.
For each path count, we report the aggregate RMSE across the nine combinations of $S_0$ and $M$, using the published
COS prices as reference values. The individual contract prices and RMSE values are provided in Appendix~\ref{app:heston-bermudan-cos}. We do not repeat the computational cost comparison here, as the relative computational efficiency of the Hybrid LSMC--PDE method and plain LSMC has already been discussed in Section~\ref{sec:computational_efficiency}.

\begin{table}[htbp]
\centering
\caption{Bermudan put pricing under the Heston specialization
of the GDMR model. Aggregate RMSE values are computed
against the published COS reference prices over nine
Bermudan contracts and $10$ independent runs. The $\mathrm{RMSE}_{\mathrm{agg}}$ ratio is the plain LSMC $\mathrm{RMSE}_{\mathrm{agg}}$ divided by the Hybrid LSMC--PDE $\mathrm{RMSE}_{\mathrm{agg}}$.}
\label{tab:heston-bermudan}
\begin{tabular}{@{}rrrr@{}}
\toprule
$N$ & Plain LSMC & Hybrid LSMC--PDE & $\mathrm{RMSE}_{\mathrm{agg}}$ ratio \\
\midrule
$5{,}000$  & 0.07422 & 0.02599 & 2.9 \\
$10{,}000$ & 0.05128 & 0.01945 & 2.6 \\
$20{,}000$ & 0.02829 & 0.01558 & 1.8 \\
\botrule
\end{tabular}
\end{table}

Consistent with the observations in Section~\ref{sec:heston-european}, the aggregate RMSE for both methods decreases as the number of simulation paths increases, as shown in Table~\ref{tab:heston-bermudan}. However the Hybrid LSMC--PDE method achieves lower aggregate RMSE than plain LSMC at all three path counts, with the RMSE reduction factors ranging from approximately 1.8 to 2.9. This observation further supports the accuracy of the hybrid method under the Heston specialization. The individual contract prices and RMSE values are reported
in Appendix~\ref{app:heston-bermudan-cos}, where the hybrid
method achieves lower RMSE than plain LSMC for all nine contracts at each path count.\par

\vspace{0.5\baselineskip}
\section{Conclusion}
\label{sec:conclusion}
In this paper, we studied the pricing of Bermudan options under a GDMR model. Bermudan pricing requires the approximation of continuation values at each exercise date, and the two variance factors in the GDMR model make this regression problem more difficult than in a one-factor model. At the same time, the model has a useful one-way coupling; the variance factors can be simulated independently of the stock price. This structure leads naturally to the Hybrid LSMC--PDE method, since the stock-price direction can be handled conditionally while the remaining variance dependence is approximated by least-squares regression. We condition on the Brownian paths driving the variance factors and project the asset Brownian motion onto these drivers. Under this conditioning, each continuation step reduces to a one-dimensional problem in log-price. This gives the Gaussian conditional representation based on $Z_n$, $I_n$, $B_n$, and the variance values at the next exercise date, and leads to an FFT-based implementation of the conditional asset step.  {As a by-product of this conditional representation, we also obtained a semi-analytical representation for European option pricing under the GDMR model.} We also established that, for the CEV exponents $\delta_1,\delta_2 \in [1/2,1]$, the GDMR variance system has a unique strong solution and remains nonnegative. This property is needed since the construction of Hybrid LSMC--PDE relies on valid simulated variance paths.

{
We conducted numerical experiments under the GDMR model using market-calibrated parameters from the literature, comparing the Hybrid LSMC--PDE method with plain LSMC using high-budget plain LSMC estimates as reference prices. The hybrid method often produced smaller pricing errors, particularly for low and moderate simulation budgets. Although the hybrid method is generally more computationally expensive for pricing a single Bermudan option, it produces an option price curve over a grid of stock prices, from which the corresponding Delta curve can be readily obtained. Additional validation under the special case of Heston model, using both market-calibrated and published parameter sets, showed good agreement with established European and Bermudan option pricing results. Moreover, conditional MC (for Euroepan options) and the Hybrid LSMC--PDE method (for Bermudan options) consistently achieved lower RMSE than plain MC and plain LSMC, respectively, across the pricing experiments considered. We note also that, for European options, both conditional MC and plain MC are computationally inexpensive.  Overall, our results suggest that the Hybrid LSMC--PDE method provides a competitive alternative to conventional Monte Carlo approaches, including plain LSMC for Bermudan options and standard Monte Carlo simulation for European options.}

For future work, we can investigate bias-control procedures such as low-biased or dual estimators, and extend the numerical study to other payoff types and early-exercise contracts.

\backmatter

\bmhead{Acknowledgements}
The authors thank Anatoliy Malyarenko for helpful discussions and the anonymous referees for valuable comments and suggestions.

\section*{Declarations}

\subsection*{Funding}
Not applicable.

\subsection*{Conflict of interest/Competing interests}
The authors declare that they have no competing interests.

\subsection*{Ethics approval and consent to participate}
Not applicable.

\subsection*{Consent for publication}
Not applicable.

\subsection*{Data availability}
No external datasets were used in this study. All numerical results were generated from simulated data based on the model specification described in the paper.
\subsection*{Materials availability}
Not applicable.

\subsection*{Code availability}
The computational code used for the numerical experiments in this study is available from the corresponding author upon reasonable request.

\subsection*{Author contributions}
MDK and YN contributed to the conceptualization, mathematical analysis, numerical design, interpretation of results, and manuscript preparation. All authors read and approved the final manuscript.
\begin{appendices}

\section{Proof of Theorem~\ref{thm:wellposed}}\label{app:proof-wellposed}

\noindent\textit{Proof.}
We first work with auxiliary equations on $\R$. For this purpose, the variance diffusion coefficients are extended by
$$
\sigma_i(x)=\xi_i(x^+)^{\delta_i},\qquad
x^+=\max\{x,0\},\qquad i=1,2.
$$
Once nonnegativity is proved, this extension is inactive because both variance factors stay in $[0,\infty)$. The auxiliary equations then coincide with the original system on the state space $\mathcal D$.

\emph{Step 1: construction of $v'$.}
Define
$$
b_2(x):=\kappa_2(\theta-x),\qquad
\sigma_2(x):=\xi_2(x^+)^{\delta_2},\qquad x\in\R.
$$
The drift $b_2$ is globally Lipschitz and has linear growth. The diffusion coefficient $\sigma_2$ is continuous, has at most linear growth, and satisfies the Yamada-Watanabe modulus condition by Assumption~\ref{ass:delta-range}. Therefore pathwise uniqueness holds for
$$
\dd v'_t=b_2(v'_t)\dd t+\sigma_2(v'_t)\dd W^{(3)}_t,
\qquad v'_0\ge0;
$$
see \cite{yamada1971uniqueness} and \cite[Ch.~5]{karatzas1991brownian}. Since the coefficients are continuous with linear growth, weak existence holds. Weak existence together with pathwise uniqueness yields a unique strong solution by the Yamada-Watanabe theorem.

It remains to show that the solution does not leave the nonnegative half-line. At the boundary,
$$
\sigma_2(0)=0,\qquad b_2(0)=\kappa_2\theta\ge0.
$$
Moreover, on the negative extension one has
$$
\sigma_2(x)=0,\qquad b_2(x)=\kappa_2(\theta-x)\ge0,\qquad x\le0.
$$
Thus the diffusion has no outward noise at the boundary and the drift points inward on the negative side. The standard one-dimensional stochastic invariance criterion for closed intervals therefore yields
$$
v'_t\ge0,\qquad 0\le t\le T,\quad \Q\text{-a.s.}
$$

\emph{Step 2: construction of $v$.}
Having constructed $v'$, define
$$
b_1(t,x):=\kappa_1(v'_t-x),\qquad
\sigma_1(x):=\xi_1(x^+)^{\delta_1}.
$$
For $m\ge1$, set
$$
\tau_m:=\inf\{t\ge0:v'_t>m\}\wedge T.
$$
Since $v'$ has continuous paths on $[0,T]$, we have $\tau_m\uparrow T$
almost surely. On $[0,\tau_m]$, the drift $b_1$ is progressively
measurable, Lipschitz in $x$ with constant $\kappa_1$, and satisfies
$$
|b_1(t,x)|\le\kappa_1(m+|x|).
$$
The diffusion coefficient $\sigma_1$ is continuous, has at most linear growth,
and satisfies the Yamada-Watanabe modulus condition. The localized
one-dimensional existence and pathwise-uniqueness theorem for SDEs with
progressive random coefficients, obtained by the same Yamada-Watanabe argument
after localization, therefore gives a unique strong stopped solution of
$$
\dd v_t=b_1(t,v_t)\dd t+\sigma_1(v_t)\dd W^{(2)}_t,
\qquad t\le\tau_m.
$$
The possible correlation between $W^{(2)}$ and $W^{(3)}$ is immaterial here:
pathwise uniqueness compares solutions driven by the same Brownian vector and
the same realized input path $v'$. By pathwise uniqueness, the stopped
solutions are compatible on overlaps. Letting $m\to\infty$ gives a unique
strong solution for $v$ on $[0,T]$.

The same invariance argument gives nonnegativity. Since $v'_t\ge0$,
$$
b_1(t,0)=\kappa_1v'_t\ge0,\qquad \sigma_1(0)=0.
$$
On the negative extension,
$$
b_1(t,x)=\kappa_1(v'_t-x)\ge0,\qquad
\sigma_1(x)=0,\qquad x\le0.
$$
Hence the closed half-line is invariant and
$$
v_t\ge0,\qquad 0\le t\le T,\quad \Q\text{-a.s.}
$$

\emph{Step 3: construction of $S$.}
The process $v$ is continuous and finite on the compact interval $[0,T]$.
Therefore
$$
\int_0^T v_s\,\dd s<\infty\qquad \Q\text{-a.s.}
$$
The asset equation is linear in $S$, and its unique strong solution is
$$
S_t
=S_0
\exp\!\left(
\int_0^t \Bigl(r-\tfrac12 v_s\Bigr)\dd s
+\int_0^t \sqrt{v_s}\,\dd W_s^{(1)}
\right).
$$
Thus $S_t>0$ for all $t\in[0,T]$ whenever $S_0>0$.

The auxiliary positive-part extension is inactive along the constructed variance
paths, because $v_t,v'_t\ge0$. Hence the constructed solution solves the
original nonnegative-state system \eqref{eq:gdmr}. Pathwise uniqueness for the
auxiliary triangular system gives pathwise uniqueness for \eqref{eq:gdmr}.

Finally, existence and pathwise uniqueness for every deterministic initial state
imply uniqueness in law. Since the coefficients are time homogeneous, the
standard strong Markov property for well-posed time-homogeneous SDEs gives that
the solution is a time-homogeneous strong Markov process on $\mathcal D$.
\hfill$\square$

\section{Pricing estimates for the numerical experiments}\label{app:price-grids}
This appendix reports the Bermudan put pricing estimates corresponding to the numerical experiments in Section~\ref{sec:numerical}. The pricing estimates reported in Tables~\ref{tab:app-step-sweep-prices-20k}--\ref{tab:app-path-sweep-60-prices} correspond to the numerical experiments in Figures~\ref{fig:step-sweep-20k}--\ref{fig:path-sweep-60}. The associated high-budget LSMC reference prices are reported in Table~\ref{tab:reference-prices}. All entries in Tables~\ref{tab:app-step-sweep-prices-20k}--\ref{tab:app-path-sweep-60-prices} are rounded to three decimal places. The label LSMC denotes the plain least-squares Monte Carlo estimator, and Hybrid denotes the Hybrid LSMC--PDE estimator.

\begin{table}[htbp]
\centering
\caption{Pricing estimates in the fixed-path step-sweep experiment with $20{,}000$ paths. Prices are rounded to three decimal places.}
\label{tab:app-step-sweep-prices-20k}
\begin{tabular}{@{}lrrrr@{}}
\toprule
Method/step & 24 & 48 & 72 & 96 \\
\midrule
LSMC $K=70$    & 2.741  & 2.612  & 2.579  & 2.536 \\
Hybrid $K=70$  & 2.711  & 2.575  & 2.547  & 2.565 \\
LSMC $K=80$    & 4.634  & 4.433  & 4.409  & 4.396 \\
Hybrid $K=80$  & 4.528  & 4.345  & 4.307  & 4.339 \\
LSMC $K=90$    & 7.341  & 7.129  & 7.068  & 7.080 \\
Hybrid $K=90$  & 7.172  & 6.957  & 6.905  & 6.956 \\
LSMC $K=100$   & 11.138 & 10.807 & 10.816 & 10.837 \\
Hybrid $K=100$ & 10.898 & 10.673 & 10.607 & 10.679 \\
LSMC $K=110$   & 16.241 & 15.920 & 15.894 & 15.798 \\
Hybrid $K=110$ & 15.953 & 15.747 & 15.676 & 15.768 \\
\botrule
\end{tabular}
\end{table}

\begin{table}[htbp]
\centering
\caption{Pricing estimates in the fixed-path step-sweep experiment with $60{,}000$ paths. Prices are rounded to three decimal places.}
\label{tab:app-step-sweep-prices-60k}
\begin{tabular}{@{}lrrrr@{}}
\toprule
Method/step & 24 & 48 & 72 & 96 \\
\midrule
LSMC $K=70$    & 2.718  & 2.612  & 2.579  & 2.533 \\
Hybrid $K=70$  & 2.721  & 2.585  & 2.557  & 2.543 \\
LSMC $K=80$    & 4.531  & 4.434  & 4.362  & 4.318 \\
Hybrid $K=80$  & 4.547  & 4.364  & 4.326  & 4.307 \\
LSMC $K=90$    & 7.195  & 7.055  & 7.033  & 6.949 \\
Hybrid $K=90$  & 7.207  & 6.987  & 6.940  & 6.918 \\
LSMC $K=100$   & 10.996 & 10.800 & 10.794 & 10.712 \\
Hybrid $K=100$ & 10.953 & 10.715 & 10.662 & 10.639 \\
LSMC $K=110$   & 16.050 & 15.788 & 15.853 & 15.806 \\
Hybrid $K=110$ & 16.026 & 15.803 & 15.749 & 15.725 \\
\botrule
\end{tabular}
\end{table}

\begin{table}[htbp]
\centering
\caption{Pricing estimates in the fixed 48-step path-sweep experiment. Prices are rounded to three decimal places.}
\label{tab:app-path-sweep-48-prices}
\setlength{\tabcolsep}{4pt}
\begin{tabular}{@{}lrrrrrrr@{}}
\toprule
Method/path count & 250 & 1,000 & 5,000 & 10,000 & 20,000 & 40,000 & 60,000 \\
\midrule
LSMC $K=70$    & 3.288  & 2.946  & 2.673  & 2.625  & 2.612  & 2.576  & 2.612 \\
Hybrid $K=70$  & 3.089  & 2.702  & 2.608  & 2.554  & 2.575  & 2.570  & 2.585 \\
LSMC $K=80$    & 5.425  & 4.795  & 4.545  & 4.452  & 4.433  & 4.369  & 4.434 \\
Hybrid $K=80$  & 5.079  & 4.533  & 4.427  & 4.310  & 4.345  & 4.357  & 4.364 \\
LSMC $K=90$    & 8.601  & 7.433  & 7.241  & 7.169  & 7.129  & 7.088  & 7.055 \\
Hybrid $K=90$  & 7.946  & 7.187  & 7.090  & 6.903  & 6.957  & 6.992  & 6.987 \\
LSMC $K=100$   & 13.454 & 11.404 & 11.072 & 10.918 & 10.807 & 10.881 & 10.800 \\
Hybrid $K=100$ & 11.870 & 10.910 & 10.833 & 10.601 & 10.673 & 10.730 & 10.715 \\
LSMC $K=110$   & 19.278 & 16.843 & 15.985 & 16.057 & 15.920 & 15.921 & 15.788 \\
Hybrid $K=110$ & 17.120 & 15.962 & 15.889 & 15.660 & 15.747 & 15.819 & 15.803 \\
\botrule
\end{tabular}
\end{table}

\begin{table}[htbp]
\centering
\caption{Pricing estimates in the fixed 60-step path-sweep experiment. Prices are rounded to three decimal places.}
\label{tab:app-path-sweep-60-prices}
\setlength{\tabcolsep}{4pt}
\begin{tabular}{@{}lrrrrrrr@{}}
\toprule
Method/path count & 250 & 1,000 & 5,000 & 10,000 & 20,000 & 40,000 & 60,000 \\
\midrule
LSMC $K=70$    & 3.452  & 2.884  & 2.738  & 2.665  & 2.635  & 2.619  & 2.587 \\
Hybrid $K=70$  & 3.125  & 2.639  & 2.653  & 2.568  & 2.563  & 2.549  & 2.570 \\
LSMC $K=80$    & 5.950  & 4.762  & 4.629  & 4.553  & 4.484  & 4.387  & 4.411 \\
Hybrid $K=80$  & 5.034  & 4.396  & 4.470  & 4.321  & 4.318  & 4.321  & 4.342 \\
LSMC $K=90$    & 8.936  & 7.227  & 7.270  & 7.222  & 7.120  & 7.045  & 7.084 \\
Hybrid $K=90$  & 7.836  & 6.957  & 7.115  & 6.910  & 6.910  & 6.939  & 6.957 \\
LSMC $K=100$   & 13.773 & 11.673 & 11.026 & 10.955 & 10.910 & 10.850 & 10.823 \\
Hybrid $K=100$ & 11.761 & 10.596 & 10.831 & 10.604 & 10.610 & 10.666 & 10.677 \\
LSMC $K=110$   & 18.962 & 17.240 & 15.920 & 15.991 & 15.918 & 15.952 & 15.860 \\
Hybrid $K=110$ & 16.953 & 15.621 & 15.863 & 15.659 & 15.680 & 15.755 & 15.757 \\
\botrule
\end{tabular}
\end{table}

\section{European put prices and RMSE in the Heston setting}
\label{app:heston-fourier}
\setcounter{table}{0}
\subsection{Fourier reference pricing formula}
\label{app:heston-fourier}

The reference prices for the European put options in
Section~6.1 are obtained using the Fourier integral
representation of Lewis~\cite{lewis2001simple} and the European put-call parity. Under the Heston setting, the European call price
is given by
\begin{equation}
C_0(S_0,K,T)
=
S_0-
\frac{e^{-rT}\sqrt{S_0K}}{\pi}
\int_0^\infty
\frac{
\operatorname{Re}\left[
e^{iu\log(S_0/K)}
\varphi_T(u-i/2)
\right]
}{
u^2+1/4
}\,du.
\label{eq:heston-lewis-call}
\end{equation}
Here $\varphi_T(u)$ is the characteristic function
of the log-return $\log(S_T/S_0)$,
\begin{equation*}
\varphi_T(u)
=
\exp\left(H_1(u,T)+H_2(u,T)v_0\right),
\label{eq:heston-cf}
\end{equation*}
where
\begin{align*}
H_1(u,T)
&=
iruT+
\frac{\kappa_1\theta}{\xi_1^2}
\left[
(\kappa_1-\rho_{12}\xi_1 iu+d(u))T
-2\log\left(
\frac{1-g(u)e^{d(u)T}}{1-g(u)}
\right)
\right],\\
H_2(u,T)
&=
\frac{\kappa_1-\rho_{12}\xi_1 iu+d(u)}{\xi_1^2}
\frac{1-e^{d(u)T}}{1-g(u)e^{d(u)T}},
\end{align*}
with
\begin{align*}
d(u)
&=
-\sqrt{
(\rho_{12}\xi_1 iu-\kappa_1)^2
-\xi_1^2(-iu-u^2)
},\\
g(u)
&=
\frac{\kappa_1-\rho_{12}\xi_1 iu+d(u)}
{\kappa_1-\rho_{12}\xi_1 iu-d(u)}.
\end{align*}

The corresponding European put reference price is
obtained from put--call parity:
\begin{equation*}
P_0(S_0,K,T)
=
C_0(S_0,K,T)+Ke^{-rT}-S_0.
\end{equation*}

The Fourier integral is evaluated numerically over $[0,\infty)$ using the \texttt{quad} function from \texttt{scipy.integrate}, with a maximum of 250 subintervals.

\subsection{Numerical results}
\label{app:heston-european}
Tables~\ref{tab:heston-european-prices} and~\ref{tab:heston-european-rmse} present detailed numerical results for a representative subset of the experiments, using 48 Euler steps per year and $N=10,000, 20,000, 40,000$. Table ~\ref{tab:heston-european-prices} reports the Fourier reference prices by the Lewis representation Eq.\ref{eq:heston-lewis-call} combined with the put-call partiy,  and the mean price estimates obtained by conditional MC (abbreviated as Cond. MC) by Eq.~\eqref{eq:european-put} in Corollary \ref{cor:european} and plain MC. Table~\ref{tab:heston-european-rmse} presents the corresponding RMSE values in option price units.

The deeply out-of-the-money put with $T=1/12$ and $K=80$ is omitted from these tables because its value is extremely small (reference price is $3.6\times 10^{-5}$). This contract is nevertheless included in the aggregate RMSE calculations reported in Section~6.1 and in the overall RMSE comparison below.

\begin{table}[htbp]
\centering
\caption{European put prices in the Heston setting.
Each simulated price is the mean of 15 independent runs. The simulations use 48 Euler steps per year.}
\label{tab:heston-european-prices}
\setlength{\tabcolsep}{4pt}
\begin{tabular}{@{}rrrrrrrrr@{}}
\toprule
 & & & \multicolumn{2}{c}{$N=10{,}000$}
     & \multicolumn{2}{c}{$N=20{,}000$}
     & \multicolumn{2}{c}{$N=40{,}000$} \\
\cmidrule(lr){4-5}\cmidrule(lr){6-7}\cmidrule(l){8-9}
$T$ & $K$ & Reference
    & Cond. MC & Plain MC
    & Cond. MC & Plain MC
    & Cond. MC & Plain MC \\
\midrule
$1/12$ &  90 &  0.030067 &  0.02737 &  0.02694 &  0.02740 &  0.02728 &  0.02741 &  0.02708 \\
$1/12$ & 100 &  1.770094 &  1.76907 &  1.76849 &  1.76932 &  1.77690 &  1.76970 &  1.77205 \\
$1/12$ & 110 &  9.850623 &  9.85039 &  9.84908 &  9.85022 &  9.85948 &  9.85132 &  9.85927 \\
$1/12$ & 120 & 19.800233 & 19.79920 & 19.79865 & 19.79897 & 19.80875 & 19.80011 & 19.80793 \\
\addlinespace[2pt]
1 &  80 &  0.525339 &  0.52544 &  0.52492 &  0.52543 &  0.52673 &  0.52582 &  0.52534 \\
1 &  90 &  2.057228 &  2.05634 &  2.05844 &  2.05593 &  2.06441 &  2.05790 &  2.06119 \\
1 & 100 &  5.547076 &  5.54441 &  5.54308 &  5.54303 &  5.55862 &  5.54776 &  5.54600 \\
1 & 110 & 11.345690 & 11.34133 & 11.34224 & 11.33882 & 11.35223 & 11.34627 & 11.33713 \\
1 & 120 & 19.053189 & 19.04776 & 19.04884 & 19.04444 & 19.05676 & 19.05360 & 19.03565 \\
\addlinespace[2pt]
2 &  80 &  1.407144 &  1.40747 &  1.40101 &  1.40783 &  1.41115 &  1.40736 &  1.42033 \\
2 &  90 &  3.576304 &  3.57681 &  3.58011 &  3.57754 &  3.57695 &  3.57684 &  3.60145 \\
2 & 100 &  7.283099 &  7.28368 &  7.30202 &  7.28490 &  7.28425 &  7.28406 &  7.32049 \\
2 & 110 & 12.602414 & 12.60295 & 12.63519 & 12.60467 & 12.59564 & 12.60377 & 12.64839 \\
2 & 120 & 19.338916 & 19.33931 & 19.38480 & 19.34144 & 19.32327 & 19.34053 & 19.38967 \\
\botrule
\end{tabular}
\end{table}

\begin{table}[htbp]
\centering
\caption{RMSE of the European put estimates against reference prices. Each entry uses the 15 individual run prices at
the indicated path count, with 48 Euler steps per year.}
\label{tab:heston-european-rmse}
\setlength{\tabcolsep}{4pt}
\begin{tabular}{@{}rrrrrrrr@{}}
\toprule
 & & \multicolumn{2}{c}{$N=10{,}000$}
   & \multicolumn{2}{c}{$N=20{,}000$}
   & \multicolumn{2}{c}{$N=40{,}000$} \\
\cmidrule(lr){3-4}\cmidrule(lr){5-6}\cmidrule(l){7-8}
$T$ & $K$
    & Cond. MC & Plain MC
    & Cond. MC & Plain MC
    & Cond. MC & Plain MC \\
\midrule
$1/12$ &  90 & 0.00271 & 0.00469 & 0.00268 & 0.00349 & 0.00266 & 0.00336 \\
$1/12$ & 100 & 0.00416 & 0.02441 & 0.00472 & 0.01846 & 0.00264 & 0.01491 \\
$1/12$ & 110 & 0.00626 & 0.03294 & 0.00728 & 0.03137 & 0.00373 & 0.02749 \\
$1/12$ & 120 & 0.00636 & 0.03419 & 0.00739 & 0.03164 & 0.00365 & 0.02825 \\
\addlinespace[2pt]
1 &  80 & 0.00550 & 0.02603 & 0.00367 & 0.01409 & 0.00251 & 0.01265 \\
1 &  90 & 0.01403 & 0.04601 & 0.00956 & 0.03575 & 0.00625 & 0.02938 \\
1 & 100 & 0.02510 & 0.07603 & 0.01754 & 0.06158 & 0.01112 & 0.04361 \\
1 & 110 & 0.03446 & 0.10097 & 0.02455 & 0.07986 & 0.01536 & 0.05369 \\
1 & 120 & 0.03990 & 0.12038 & 0.02879 & 0.08436 & 0.01791 & 0.06395 \\
\addlinespace[2pt]
2 &  80 & 0.00977 & 0.06105 & 0.00712 & 0.02372 & 0.00532 & 0.02233 \\
2 &  90 & 0.01815 & 0.10590 & 0.01304 & 0.04041 & 0.01040 & 0.03355 \\
2 & 100 & 0.02749 & 0.15352 & 0.01964 & 0.05840 & 0.01638 & 0.04864 \\
2 & 110 & 0.03600 & 0.20093 & 0.02567 & 0.07717 & 0.02205 & 0.06429 \\
2 & 120 & 0.04259 & 0.23555 & 0.03037 & 0.09829 & 0.02656 & 0.07636 \\
\botrule
\end{tabular}
\end{table}

Table~\ref{tab:heston-european-rmse} shows that conditional MC yields lower RMSE than plain MC for all reported contracts and path simulation budgets with 48 Euler steps per year. Indeed, our experiments show that, across all 180 configurations (3 maturities, 5 strikes, 3 time discretizations, and 4 path simulation budgets), conditional MC achieved lower RMSE than plain MC in 178 cases. The only two exceptions occurred for the short-maturity out-of-the-money put ($T=1/12$, $K=90$) with 24 steps per year. However the differences in RMSE in these two exceptional case were very small.

\section{Bermudan put prices and RMSE in the Heston setting}
\label{app:heston-bermudan-cos}
\setcounter{table}{0}
For the comparison in Section~\ref{sec:heston-bermudan}, both methods use $120$ Euler steps over the full maturity, giving six, three, and two steps per exercise interval for $M=20,40,60$, respectively. The hybrid method uses $N_S=301$ grid points uniformly spaced in log-price, as in Section~\ref{sec:numerical}. The mean prices and their corresponding COS reference prices are presented in Table~\ref{tab:heston-bermudan-cos-prices}, while the contract-specific RMSE values are reported in
Table~\ref{tab:heston-bermudan-cos-rmse}. In both tables, Hybrid denotes the Hybrid LSMC--PDE estimator and LSMC denotes the plain LSMC estimator.

\begin{table}[htbp]
\centering
\caption{Bermudan put prices in the Heston setting. The COS column gives the published reference prices of Fang and Oosterlee~\cite{FangOosterlee2011}. Each simulated price is the mean of $10$ runs.}
\label{tab:heston-bermudan-cos-prices}
\setlength{\tabcolsep}{4pt}
\begin{tabular}{@{}rrrrrrrrr@{}}
\toprule
 & & & \multicolumn{2}{c}{$N=5{,}000$}
     & \multicolumn{2}{c}{$N=10{,}000$}
     & \multicolumn{2}{c}{$N=20{,}000$} \\
\cmidrule(lr){4-5}\cmidrule(lr){6-7}\cmidrule(l){8-9}
$M$ & $S_0$ & COS & Hybrid & LSMC & Hybrid & LSMC & Hybrid & LSMC \\
\midrule
20 &  90 & 9.9783714 &  9.98950 & 10.02172 &  9.98583 & 10.00121 &  9.98449 &  9.99168 \\
20 & 100 & 3.2047434 &  3.20935 &  3.22294 &  3.19682 &  3.21223 &  3.19669 &  3.19976 \\
20 & 110 & 0.9273568 &  0.92665 &  0.94492 &  0.92143 &  0.92521 &  0.92189 &  0.92329 \\
\addlinespace[2pt]
40 &  90 & 9.9916484 & 10.00384 & 10.05426 &  9.99999 & 10.02633 &  9.99746 & 10.00897 \\
40 & 100 & 3.2073345 &  3.21132 &  3.21576 &  3.19818 &  3.22073 &  3.19841 &  3.20083 \\
40 & 110 & 0.9281068 &  0.92846 &  0.93694 &  0.92228 &  0.92240 &  0.92332 &  0.91678 \\
\addlinespace[2pt]
60 &  90 & 9.9957789 & 10.01026 & 10.06164 & 10.00350 & 10.02918 & 10.00176 & 10.01750 \\
60 & 100 & 3.2079202 &  3.21243 &  3.21471 &  3.19866 &  3.21632 &  3.19844 &  3.19363 \\
60 & 110 & 0.9280425 &  0.92898 &  0.93015 &  0.92286 &  0.91943 &  0.92376 &  0.91629 \\
\botrule
\end{tabular}
\end{table}

\begin{table}[htbp]
\centering
\caption{RMSE of the Bermudan put estimates against the corresponding COS prices, in option price units. Each entry uses the $10$ individual run prices at the indicated path count.}
\label{tab:heston-bermudan-cos-rmse}
\setlength{\tabcolsep}{4pt}
\begin{tabular}{@{}rrrrrrrr@{}}
\toprule
 & & \multicolumn{2}{c}{$N=5{,}000$}
   & \multicolumn{2}{c}{$N=10{,}000$}
   & \multicolumn{2}{c}{$N=20{,}000$} \\
\cmidrule(lr){3-4}\cmidrule(lr){5-6}\cmidrule(l){7-8}
$M$ & $S_0$ & Hybrid & LSMC & Hybrid & LSMC & Hybrid & LSMC \\
\midrule
20 &  90 & 0.02534 & 0.06183 & 0.01523 & 0.04874 & 0.00973 & 0.02608 \\
20 & 100 & 0.03340 & 0.08996 & 0.02588 & 0.06852 & 0.02127 & 0.03697 \\
20 & 110 & 0.01958 & 0.04883 & 0.01581 & 0.03431 & 0.01350 & 0.01940 \\
\addlinespace[2pt]
40 &  90 & 0.02302 & 0.08708 & 0.01455 & 0.04874 & 0.00899 & 0.02628 \\
40 & 100 & 0.03301 & 0.08985 & 0.02642 & 0.06783 & 0.02151 & 0.03534 \\
40 & 110 & 0.01908 & 0.04811 & 0.01602 & 0.03580 & 0.01330 & 0.02230 \\
\addlinespace[2pt]
60 &  90 & 0.02189 & 0.07515 & 0.01327 & 0.04164 & 0.00879 & 0.02560 \\
60 & 100 & 0.03315 & 0.09613 & 0.02601 & 0.06159 & 0.02202 & 0.03561 \\
60 & 110 & 0.01945 & 0.04917 & 0.01535 & 0.04039 & 0.01319 & 0.02050 \\
\botrule
\end{tabular}
\end{table}

Table~\ref{tab:heston-bermudan-cos-rmse} and Figure~\ref{fig:app-heston-cos-rmse} show that RMSE decreases at each doubling of the path count for both methods and all nine combinations. In particular the Hybrid LSMC--PDE method has the smaller RMSE in all 27 comparisons. At $20{,}000$ paths, RMSE ranges from $0.00879$ to $0.02202$ for the hybrid method and from $0.01940$ to $0.03697$ for the plain LSMC.
\begin{figure}[H]
\centering
\includegraphics[width=\linewidth]{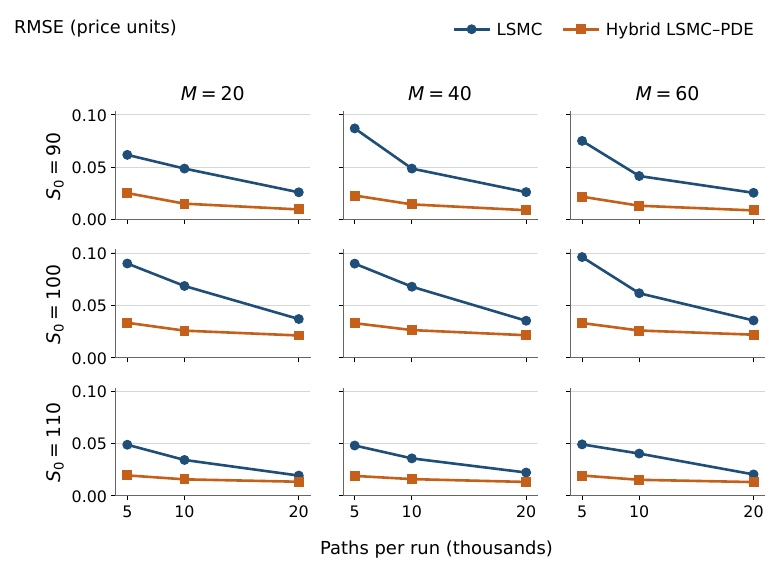}
\caption{\textcolor{blue}{Bermudan put RMSE against the published Heston COS prices.}}
\label{fig:app-heston-cos-rmse}
\end{figure}

\end{appendices}
\bibliography{references}

\end{document}